\documentclass[11pt]{article}

\usepackage{makeidx}
\usepackage{times}
\usepackage{amssymb}
\usepackage{amsmath}
\numberwithin{equation}{section}
\usepackage{comment}
\usepackage{color}
\usepackage{graphicx}
\usepackage{rotating}
\usepackage{pdflscape}
\usepackage{epstopdf}
\usepackage[round,authoryear,comma]{natbib}
\usepackage{natbib,hyperref}
\usepackage{booktabs}
\usepackage{setspace}
\usepackage{pgfplots}
\usepackage{subcaption}
\usepackage{xcolor}
\usepackage{multirow}
\hypersetup{
	colorlinks   = true, 
	urlcolor     = blue, 
	linkcolor    = blue, 
	citecolor   = blue 
}

\renewcommand{\vec}[1]{\boldsymbol{#1}} 
\newcommand\numberthis{\addtocounter{equation}{1}\tag{\theequation}}

\providecommand{\U}[1]{\protect\rule{.1in}{.1in}}
\newtheorem{theorem}{Theorem}[section]

\newtheorem{corollary}{Corollary}[section]

\newtheorem{definition}{Definition}[section]
\newtheorem{example}{Example}

\newtheorem{lemma}{Lemma}[section]

\newtheorem{proposition}{Proposition}
\newtheorem{remark}[theorem]{Remark}

\newenvironment{proof}[1][Proof]{\noindent \textbf{#1.} }{\hfill
	\rule{0.5em}{0.5em}}

\parskip=\medskipamount

\definecolor{alizarin}{rgb}{0.82, 0.1, 0.26}

\newcommand{\blue}{\color{black}}
\newcommand{\black}{\color{black}}
\newcommand{\bluebis}[1]{\textcolor{black}{#1}}
\newcommand{\myparagraph}[1]{\paragraph{#1}\mbox{}\\}

\pgfplotsset{compat=1.17} 

\begin{document}

\title{\textbf{Actuarial-consistency and two-step actuarial valuations: a new paradigm to insurance valuation }}
\date{Version: \today}
\author{Karim Barigou\footnote{Univ Lyon, Université Claude Bernard Lyon 1,
Laboratoire de Sciences Actuarielle et Financière,
Institut de Science Financière et d’Assurances
(50 Avenue Tony Garnier, F-69007 Lyon, France). E-mail address: \href{mailto:karim.barigou@univ-lyon1.fr}{karim.barigou@univ-lyon1.fr}} 
\and
Dani\"{e}l Linders\footnote{University of Amsterdan, Netherlands. E-mail address: \href{mailto:dlinders@illinois.edu}{d.h.linders@uva.nl}}
\and Fan Yang\footnote{University of Waterloo, Waterloo, ON N2L 3G1, Canada. E-mail address: \href{mailto:fan.yang@uwaterloo.ca}{fan.yang@uwaterloo.ca}} 
}
\maketitle

\begin{abstract}
This paper introduces new valuation schemes called actuarial-consistent valuations for insurance liabilities which depend on both financial and actuarial risks, which  imposes that all actuarial risks are priced via standard actuarial principles. We propose to extend standard actuarial principles by a new actuarial-consistent procedure, which we call ``two-step actuarial valuations". \bluebis{In the case valuations are coherent}, we show that actuarial-consistent valuations are equivalent to two-step actuarial valuations. We also discuss the connection with ``two-step market-consistent valuations" from \cite{pelsser2014time}. In particular, we discuss how the dependence structure between actuarial and financial risks impacts both actuarial-consistent and market-consistent valuations. 
\medskip


\end{abstract}

\section{Introduction}

Insurance liabilities, such as variable annuities, are complex combinations of different types of risks. Motivated by solvency regulations, the recent focus has been towards financial risks and the so-called market-consistent valuations. In this situation, the financial market is the main driver and actuarial risks only appear as the “second step”. This paper goes against the tide and introduces the concept of \textit{actuarial-consistent} valuations where actuarial risks are at the core of the valuation. We propose a two-step actuarial valuation that is first driven by actuarial information and is actuarial-consistent. Moreover, we show that actuarial-consistent valuations can always be expressed as the price of an appropriate hedging strategy.

Fundamental to insurance, an actuarial premium principle (called \textit{actuarial valuation} in this paper) is typically based on a diversification argument which
justifies applying the law of large numbers (LLN) among independent policyholders who face identical risks (see \citealp{denuit2006actuarial}). \blue Consistent with insurance regulation, the actuarial valuation $\rho(S)$ of a claim $S$ is represented as the expectation under the real-world probability measure $\mathbb{P}$ plus an additional risk margin to cover any undiversified and systematic risk\black, that is
\begin{equation}
    \rho (S) = \mathbb{E} ^{\mathbb{P}}\left[ S \right] + \text{Risk margin.}\label{actuarialvaluation}
\end{equation}

Financial valuation, on the other hand, is based on the no-arbitrage principal. Given the prices of the available traded assets, the value of a financial claim should be determined such that the market remains free of arbitrage if the claim is traded. Therefore, financial pricing is based on the idea of hedging and replication. It  was shown in  \cite{delbaen2006mathematics} that no-arbitrage pricing implies that the prices of contingent claims can be expressed as expectations under a so-called risk-neutral measure $\mathbb{Q}$, that is
\begin{equation}
    \rho (S) = \mathbb{E} ^{\mathbb{Q}}\left[ S \right]. \label{financialvaluation}
\end{equation}
This approach dates back to the seminal paper of \cite{black1973pricing}.

Insurance claims are nowadays non-trivial combinations of \blue diversifiable and undiversifiable insurance risk, and traded financial risks. \black It is therefore primordial to build a valuation framework which combines traditional actuarial and financial valuation. \bluebis{Over the last two decades, several researchers have worked on the interplay between financial and actuarial valuation. \cite{Embrechts2000ActuarialVF} offers a detailed comparison of insurance and finance pricing mechanisms. \cite{moller2002valuation} addresses the aspects of the interplay between finance and insurance by combining traditional actuarial and financial pricing principles.} A simplifying approach in the actuarial literature is to assume independence between actuarial and financial risks\footnote{This assumption is either made under the real-world measure $\mathbb{P}$ or the risk-neutral measure $\mathbb{Q}$. We note that the independence under $\mathbb{P}$ does not necessarily imply the independence under $\mathbb{Q}$; see \cite{dhaene2013dependence}. We also discuss this point in Lemma \ref{lemmaind}.} such that the valuation can be split into a product of actuarial and financial valuations (see \citealp{fung2014systematic}; \citealp{da2017valuing}; \citealp{ignatieva2016pricing}; \citealp{Wuthrich_2016} among others). However, the emergence of longevity-linked financial products and the pandemic situation showed us that future mortality cannot be assumed independent from evolution of the financial market \citep{sharif2020covid,harjoto2021equity}. For this reason, different authors proposed general valuation approaches allowing for dependencies between actuarial and financial risks. \bluebis{For instance, valuation under dependent mortality and interest risks was investigated in \cite{liu2014generalized,deelstra2016role,zhao2018efficient}.} Moreover, \cite{pelsser2014time} proposed a `two-step market-consistent valuation' which extends standard actuarial principles by conditioning on the financial information. \cite{dhaene2017} proposed a new framework for the fair valuation of insurance liabilities in a one-period setting; see also \cite{Dhaene_2021}. The authors introduced the notion of a `fair valuation', which they defined as a valuation which is both market-consistent (mark-to-market for any hedgeable part of a claim) and actuarial (mark-to-model for any claim that is independent of financial market evolution). This work was further extended in a multi-period discrete setting in \cite{barigoumultiperiod} and in continuous time in \cite{delong2019fair}.  A  3-step valuation was introduced in \cite{deelstra_devolder_gnameho_hieber_2020} for the valuation of claims which consists of traded, financial but also systematic risks. This approach was further generalized  in \cite{Linders_2021}.


\blue
\myparagraph{Market-consistent valuation and its shortfalls}
The above-mentioned papers propose different valuation principles which  have in common that they are all market-consistent valuations. The Solvency II insurance regulation directive introduced a prospective and risk-based supervisory approach on January 1, 2016. Pillar 1 of this directive requires a market-consistent valuation of the insurance liabilities; see e.g.\ \cite{mohr_2011}. A market-consistent valuation assumes an investment in an appropriate replicating portfolio to offset the hedgeable part of the liability. The remaining part of the claim is managed by diversification and an appropriate capital buffer. However, concerns about the appropriateness of the market-consistent valuation for long-term insurance business were raised: 
\begin{itemize}
    \item \cite{lecourtois2021utility} pointed out market-consistency can lead to a substantial misvaluation of the effective wealth of an insurance company that deals with long-term commitments. This approach also induces high instability and excess volatility in the balance sheet indicators of an insurance company \citep{vedani2017market,rae2018review}.
    \item \cite{plantin2008marking} found that the damage done by marking-to-market (that is market-consistency) is greatest when claims are illiquid and long term, which is precisely the case of balance sheet of insurance companies. 
    \item Market-consistency tends to be pro-cyclical and the use of a 1-year Value-at-Risk increases the risk of herd behaviour, hence reducing financial stability. Market-consistent valuation tends to minimise the value of liabilities when markets are bullish and over-estimate them in times of crisis \citep{rae2018review}.
\end{itemize}
\myparagraph{Alternative valuations to market-consistency}
The second pillar of the Solvency II directive allows the insurer to proceed to an alternative assessment of the company’s overall solvency needs using different recognition and valuation bases.
In this purpose, \cite{lecourtois2021utility} replaced the market-consistent framework by a utility-consistent framework that accounts for the risk aversion of the market and the long-term nature of liabilities. The framework is not anymore market-consistent but market-implied and utility-consistent. The authors found that this alternative valuation provides less volatility than the traditional market-consistent approach. \cite{muermann2003actuarially} investigates the valuation of catastrophe derivatives and proposes an actuarial-consistent valuation approach that is consistent with existing insurance premiums to exclude arbitrage opportunities.
\black

\myparagraph{Why actuarial-consistency is an important alternative ?}
In the spirit of \cite{lecourtois2021utility,muermann2003actuarially}, we introduce the class of actuarial-consistent valuations for hybrid claims as an alternative for the market-consistent valuations. This new valuation framework is motivated by the requirement that any actuarial claim \blue(such as a pure endowment) \black should be priced via an actuarial valuation and should not be managed using a risky investment. We will label this property of the valuation ``actuarial-consistency". \blue Such property plays an important role in life insurance where the development of the longevity market and the longevity-linked securities is growing in the recent years \citep{blake2020longevity}. Indeed, in several countries, the valuation of the longevity risk is often required to follow regulatory life tables that impose minimum loadings for actuarial claims. Under a market-consistent setting, the presence of undervalued longevity-linked securities might create a potential undervaluation of life insurance business. Therefore, we believe it is relevant, for certain insurance claims, to investigate valuation mechanisms that are not market-consistent and would encourage risk-free investments rather than longevity transfers to the capital market.\footnote{As pointed out in \cite{blake2020longevity}, the Prudential Regulatory Authority (regulatory authority for insurance companies in the UK) expressed concerns that too much longevity risk is transferred offshore so that if the offshore reinsurance firm failed, UK pensioners might not get their pensions.} \black

We also introduce the two-step actuarial valuations. Instead of first considering the hedgeable part of the claim, the two-step actuarial valuation will first price the actuarial part of the claim using an actuarial valuation. We show that every two-step actuarial valuation is actuarial-consistent. Moreover, \bluebis{if the valuation is coherent}, we show the reciprocal: any actuarial-consistent valuation has a two-step actuarial representation. The two-step valuations are general in the sense that they do not impose linearity constraints on the actuarial and financial valuations. Therefore, they allow to account for the diversification of  actuarial risks and/or the incompleteness of the financial market (e.g.\ non-linear financial pricing with bid-ask prices).

Hedge-based valuations were first introduced in \cite{dhaene2017} to define the market-consistent valuations. We show that actuarial-consistent valuations can always be expressed as hedge-based valuations. The hedging strategy used in an actuarial-consistent valuation will only invest in the risk-free asset when valuating an actuarial claim. This is in contrast with the market-consistent hedge-based valuations, which may use the financial market to hedge actuarial claims. 

The paper also provides a detailed comparison between two-step market and two-step actuarial valuations. We discuss how the dependence structure between actuarial and financial claims impacts both actuarial and market-consistent valuations. In the context of solvency regulations, we show how the two-step actuarial valuation can be decomposed into a best estimate (expected value) plus a risk margin to cover the uncertainty in the actuarial risks. The procedure will be illustrated with a portfolio of life insurance contracts with dependent financial and actuarial risks.

The rest of the paper is structured as follows. In Section \ref{sectionTS}, we describe financial and actuarial valuations. Section \ref{ACvaluations} discusses the notion of actuarial-consistency and introduces two-step actuarial valuations. In Section \ref{sectionFV}, we provide a detailed comparison between actuarial-consistent valuations and market-consistent valuations. Section \ref{sectionNI} presents a cost-of-capital valuation based on the two-step actuarial valuation and a detailed numerical application of the two-step actuarial valuation on a portfolio of equity-linked contracts. Section \ref{Conclu4} concludes the paper. 

\section{Actuarial and financial valuations}

\label{sectionTS}

All random variables introduced hereafter are defined on the
probability space $\left(  \Omega,\mathcal{F},\mathbb{P}\right)  $. Equalities
and inequalities between r.v.'s have to be understood in the $\mathbb{P}%
$-almost sure sense. \blue The space of bounded random variables is denoted by
$L^{\infty}(\Omega,\mathcal{F},\mathbb{P})$ or $L^{\infty}(\mathcal{F})$ for
short. \black A contingent claim is a random liability of an insurance company
that has to be paid at the deterministic future time $T$. Formally, a
discounted claim is modeled by the random variable $S\in L^{\infty
}(\mathcal{F})$. In what follows we are interested in the \emph{valuation} of
discounted claims.

Suppose a $\sigma$-algebra $\mathcal{G}\subset \mathcal{F}$ is the information
available to the agent. We define a $\mathcal{G}$-conditional valuation as follows.
\blue
\begin{definition}
[$\mathcal{G}$-conditional valuation]\label{def}A $\mathcal{G}$-conditional
valuation is a mapping $\Pi \left[  \cdot|\mathcal{G}\right]  :L^{\infty
}(\mathcal{F})\longrightarrow L^{\infty}(\mathcal{G})$ satisfying the
following properties:
\begin{itemize}
\item \underline{Normalization}: $\Pi \lbrack0|\mathcal{G}]=0$.

\item \underline{Translation-invariance}: For any $S\in L^{\infty}%
(\mathcal{F})$ and $\lambda \in L^{\infty}(\mathcal{G})$, we have
\[
\Pi \left[  S+\lambda|\mathcal{G}\right]  =\Pi \lbrack S|\mathcal{G}]+\lambda.
\]

\item \underline{Convexity}: For any $S_1,S_2\in L^{\infty}(\mathcal{F})$ and $\lambda \in L^{\infty}(\mathcal{G})$ with $ 0\leq \lambda \leq 1$, we have
\[
\Pi \left[ \lambda S_1+(1-\lambda)S_2|\mathcal{G}\right]  \leq \lambda \Pi \lbrack S_1|\mathcal{G}]+ (1-\lambda)\Pi \lbrack S_2|\mathcal{G}].
\]

\item \underline{Positive homogeneity}: For any $S\in L^{\infty}(\mathcal{F})$
and any positive $\lambda \in L^{\infty}(\mathcal{G})$, we have
\[
\Pi \left[  \lambda S|\mathcal{G}\right]  =\lambda \Pi \left[  S|\mathcal{G}%
\right].
\]

\end{itemize}
\end{definition}\black

Note that if $\mathcal{G}$ is chosen to be the trivial $\sigma$-algebra, then
$\Pi \left[  \cdot|\mathcal{G}\right]  $ is a real number and we simply write the valuation $\Pi
\lbrack \cdot]$ without conditioning on $\mathcal{G}$.

Apart from the properties in Definition \ref{def}, other properties that
a valuation may have are

\begin{itemize}
\item \underline{Monotonicity}: for any $S,R\in L^{\infty}(\mathcal{F})$,
$\Pi(S)\leq \Pi(R)$ if $S\leq R$.

\item \underline{Fatou property}: for $S,S_{1},S_{2},...\in L^{\infty}(\mathcal{F})$,
$\sup_{n\in \mathbb{N}}\left \Vert S_{n}\right \Vert <\infty$, and $S_{n}%
\overset{a.s.}{\rightarrow}S$, then
\[
\underset{n\rightarrow \infty}{\lim \inf}\Pi(S_{n})\geq \Pi \lbrack S].
\]

\end{itemize}

Coherent valuations (also known as coherent risk measures) can be
understood as a worst-case expectation with respect to some class of
probability measures. This can be motivated by the desire for robustness: the
valuator does not only want to rely on a single measure $\mathbb{P}$ for the
occurrence of future events but prefers to test a set of plausible measures
and value with the worst-case scenario. A few dual representations of coherent
valuations are available in the literature. Here we present the most popular
version, which is defined for bounded random variables on a general space; see e.g. \cite{delbaen2000coherent}.
\blue
\begin{definition}
[Coherent valuation]\label{dual} A coherent valuation with Fatou property $\Pi:L^{\infty}(\mathcal{F}%
)\rightarrow \mathbb{R}$ has the following representation
 for any $S\in L^{\infty}(\mathcal{F})$
\[
\Pi \lbrack S]=\sup_{Q\in \mathcal{M}}\mathbb{E}^{Q}\left[  S\right]  ,
\]
where $\mathcal{M}$ is a collection of probability measures absolutely
continuous w.r.t. $\mathbb{P}$.
\end{definition}\black

In general, one can show that the set of coherent valuations is positive
homogeneous, translation-invariant, subadditive and monotone, see e.g.
\cite{artzner1999coherent}.

Coherent valuations are linear valuations in the following sense.
\blue
\begin{definition}
[Linear valuation]A valuation $\Pi:L^{\infty}(\mathcal{F})\rightarrow
\mathbb{R}$ is linear, if there exists a finitely additive measure $Q$
absolutely continuous w.r.t. $\mathbb{P}$, $Q(\Omega)=1$, such that for any
$S\in L^{\infty}(\mathcal{F})$
\[
\Pi \left[  S\right]  =\mathbb{E}^{Q}\left[  S\right]  .
\]

\end{definition}\black

In this paper, we assume risks can be divided in two groups: financial and
actuarial risks. Then a hybrid discounted claim $S\in L^{\infty}(\mathcal{F})$
is a combination of both the actuarial as well as the financial risks. The
financial risks are traded on a public exchange and market participants can
buy and sell these financial risks at any quantity. The non-traded risks are
referred to as actuarial risks. We note that a similar split was considered in
\cite{Wuthrich_2016} between financial events (stocks, asset portfolio,
inflation-protected bonds, etc) and insurance technical events (death, car
accident, medical expenses, etc). Moreover, we remark that the issue of
insurance-linked securities (e.g.\ longevity bonds) implies that a non-traded
actuarial risk may become a traded financial risk.

\subsection{Financial Valuation}\label{FV}
\bluebis{We assume there is a financial market with $n^{(1)}+1$ traded assets and
denote by $\vec{Y}=(Y_{0},Y_{1},\ldots,Y_{n^{(1)}})$ the price process of
assets, where $Y_{0}$ is the risk-free asset assumed to be constant and equal
to $1$.  We also refer to the traded assets as financial risks.
The $\sigma$-algebra generated by the financial risks $\vec{Y}$ is denoted by
$\mathcal{F}^{(1)}$.}

A \emph{financial claim} is an $\mathcal{F}^{(1)}%
$-measurable random variable defined on the financial probability space.
Otherwise stated, a financial claim only depends on the financial risks
$\vec{Y}$ and its realization is completely known given the realization of the
financial risks $\vec{Y}$. The set of all discounted financial claims is
denoted by $L^{\infty}(\mathcal{F}^{(1)})$. We can always express a discounted
financial claim as a function of the financial risks. We have:
\begin{equation}
S^{(1)}=f\left(  \vec{Y}\right)  ,\  \text{if}\ S^{(1)}\in L^{\infty
}(\mathcal{F}^{(1)}),\label{FinClaim}%
\end{equation}
for some function $f$\footnote{We assume that all functions we encounter are
Borel measurable.}.

The financial risks are traded on the financial market and all market
participants can observe the prices at which each asset can be bought and
sold. Note, however, that we do not assume that the price at which one can buy
and sell at time $t$ is equal. We assume that a financial valuation principle
$\pi^{(1)}$ is available to determine the price at which one can buy the
traded payoffs:
\[
\pi^{(1)}:L^{\infty}(\mathcal{F}^{(1)})\rightarrow \mathbb{R}.
\]

The choice of the financial valuation principle $\pi^{(1)}$ depends on the
additional assumptions we impose on the financial market. Traditional
financial pricing assumes that the market is complete and arbitrage-free such
that the pricing rule is linear and unique \citep{black1973pricing}. However,
the existence of transaction costs and non-hedgeable payoffs lead to
non-linear and non-unique valuations. In such situations, the market will
decide which valuation principle is used and one has to use calibration to
back out the valuation principle from the available traded assets. Below, we
consider different possible choices for the financial valuation for different
market situations.
\begin{enumerate}
	\item \underline{The law of one price.}
	Assume the market is arbitrage-free, frictionless and that all financial assets are discretely traded in the market and can be bought and sold at a unique price. One can prove that in this market setting, the no-arbitrage condition is equivalent with the existence of an equivalent martingale measure (EMM) $\mathbb{Q}$ \citep{dalang1990equivalent}.
	In this financial market where one can buy and sell any asset at a unique price, the financial valuation principle can be determined as follows:
	$$\pi^{(1)}[S]=\mathbb{E}^{\mathbb{Q}}[S].$$
	 The financial valuation is in this situation a linear valuation principle. 
	
	\item \underline{Imperfect market and bid-ask prices.}
	In classical finance, markets are usually modelled as a counterparty for market participants. It is assumed that markets can accept any amount and direction of the trade (buy or sell) at the going market price. However, due to market imperfection, there is in practice a difference between the price the market is willing to buy (bid price) and the price the market is willing to sell (ask price). This difference, called the bid-ask spread, creates a two-price economy. In particular, the value $\pi^{(1)}[S]$ which corresponds with the price required by the market to take over the financial claim $S$ will typically be higher than the risk-neutral price. Indeed, the asymmetry in the market allows that market to take a more prudent approach when determining the price $\pi^{(1)}[S]$. Instead of using a single risk-neutral probability measure, a set of "stress-test measures" is selected from the set of martingale measures and the price is determined as the supremum of the expectations w.r.t. the stress-test measures:
	$$\pi^{(1)}[S]=\sup_{\mathbb{Q}\in\mathcal{Q}} \mathbb{E}^\mathbb{Q} \left[S \right],$$
	where $\mathcal{Q}$ is a convex set of  probability measures absolutely continuous with respect to the original probability. For more details on conic finance, we refer to \cite{madan2010markets} and \cite{madan2016applied}. 
\end{enumerate}

In the remainder of the paper, we consider coherent financial valuations to account for bid-ask spread and market imperfections. 
\subsection{Actuarial valuation}

\bluebis{Suppose there are $n^{(2)}$ non-traded risks (also called actuarial
risks) which we denote by $\vec{X}=(X_{1},\ldots,X_{n^{(2)}})$ the price
process of the actuarial risks. The $\sigma$-algebra generated by $\vec{X}$ is
denoted by $\mathcal{F}^{(2)}$.} 

An \emph{actuarial claim} is a $\mathcal{F}^{(2)}$-measurable random variable
defined on the actuarial probability space. Equivalently stated, an actuarial
claim only depends on $\vec{X}$ and its realization is completely known given
the realization of the actuarial risks $\vec{X}$. If we denote the set of
discounted actuarial claims by $L^{\infty}(\mathcal{F}^{(2)})$, we can write
any discounted claim $S^{(2)}\in L^{\infty}(\mathcal{F}^{(2)})$ as a function of the
actuarial risks:
\[
S^{(2)}=f\left(  \vec{X}\right)  ,\  \text{if}\ S^{(2)}\in L^{\infty
}(\mathcal{F}^{(2)}),
\]
for some function $f$.

We assume that a valuation principle $\pi^{(2)}$ is chosen to price actuarial
claims. The actuarial valuation principle $\pi^{(2)}$ is based on the idea of
pooling and diversification. A completely diversifiable portfolio
can be valued with its expectation under the physical measure $\mathbb{P}$.
However, there is always an amount of residual actuarial risk present because one cannot aggregate an infinite amount of policies. Moreover, there are also
systematic actuarial risks (e.g.\ longevity risk) which cannot be diversified
away. Indeed, aggregating a large number of systematic risks will not result in the desired risk reduction.
\bluebis{For this reason, actuarial valuations include a risk margin to cover any non-diversifiable risk.}
Below, we briefly discuss the most important actuarial valuation principles
(also called actuarial premium principles).
\begin{enumerate}
	\item \underline{Linear valuation:}
	$$\pi^{(2)}[S]=\mathbb{E}^{\tilde{\mathbb{P}}}\left[ S\right].$$
	The risk margin is modelled by an appropriate change of measure from $\mathbb{P}$ to $\tilde{\mathbb{P}}$. In terms of life tables, the change of measure can be interpreted as a switch from the second order life table  (best-estimate survival or death probabilities) to a first order life table (survival or death probabilities that are chosen with a safety margin). For more details, see for instance \cite{Wuthrich_2016} and \cite{norberg2014life}. 
	\item \underline{Standard deviation principle:}
	\begin{equation}\label{stdprinciple}
	\pi^{(2)}[S]=\mathbb{E}^{\mathbb{P}}\left[ S\right]+\beta\sqrt{\text{Var}^{\mathbb{P}}\left[ S\right]},
	\end{equation}
	with $\beta\geq0 $. \newline In this case, the loading equals $\beta$ times the standard deviation. It is well-known that $\beta>0$ is required in order to avoid getting ruin with probability 1 (see e.g. \cite{kaas2008modern}). We note that the standard deviation principle is neither linear nor coherent.
	\item \underline{Coherent valuation:}
	$$\pi^{(2)}[S]=\rho\left[ S\right],$$
	where $\rho$ is a coherent valuation. We
	recall that the coherent valuation can also be represented as a supremum of a set of measures (see Definition \ref{dual}). Therefore, model risk can be taken into account by considering a family of different distributions and the actuarial claim is valuated with the most conservative one. 
\end{enumerate}
%
%
%

\subsection{Hybrid claims}
\bluebis{Recall that a claim $S\in L^{\infty}(\mathcal{F})$ is defined on the
probability space $\left(  \Omega,\mathcal{F},\mathbb{P}\right)  $, which is a
combination of the financial and the actuarial probability space. The $\sigma
$-algebra $\mathcal{F}$ contains $\mathcal{F}^{(1)}$ and $\mathcal{F}^{(2)}$.
That is $\left(  \mathcal{F}^{(1)}\cup \mathcal{F}^{(2)}\right)  \subset
\mathcal{F}$.}

Assume the valuation $\Pi:L^{\infty}(\mathcal{F})\rightarrow \mathbb{R}$ is
used to valuate discounted claims in $L^{\infty}(\mathcal{F})$. A general
claim in $L^{\infty}(\mathcal{F})$ can contain both financial and actuarial
risks and therefore the valuation of claims in $L^{\infty}(\mathcal{F})$
cannot be solely based on the financial or the actuarial valuation principles.
\bluebis{Moreover, the financial and actuarial risks are dependent.}
Therefore, observing the values of financial (resp.\ actuarial) claims can
provide information about the valuation of actuarial (resp.\ financial) claims.

Define by $L^{\infty}(\mathcal{F}^{(1,\perp)})$ the set of financial claims
which are independent of the actuarial risks and by $L^{\infty}(\mathcal{F}%
^{(2,\perp)})$ the actuarial claims which are independent of the financial
information:
\begin{align*}
S^{(1,\perp)}\in L^{\infty}(\mathcal{F}^{(1,\perp)}) &  \text{if}%
\ S^{(1,\perp)}\in L^{\infty}(\mathcal{F}^{(1)})\text{ and}\ S^{(1,\perp
)}\perp \vec{X},\\
S^{(2,\perp)}\in L^{\infty}(\mathcal{F}^{(2,\perp)}) &  \text{if}%
\ S^{(2,\perp)}\in L^{\infty}(\mathcal{F}^{(2)})\text{ and}\ S^{(2,\perp
)}\perp \vec{Y}.
\end{align*}
We say that $S^{\left(  1,\perp \right)  }$ is a pure financial claim, whereas
$S^{\left(  2,\perp \right)  }$ is called a pure actuarial claim. A pure
financial claim does not contain any information about the actuarial risks and
therefore the actuarial valuation $\pi^{(2)}$ should not be used to valuate a
pure financial claim. Hence, the valuation of a pure financial claim should
only involve the financial valuation $\pi^{(1)}$. Similarly, the valuation of
a pure actuarial claim should only be based on the actuarial valuation
$\pi^{(2)}$. We require that the valuation principle $\Pi$ is consistent with
the financial valuation $\pi^{(1)}$ and the actuarial principle $\pi^{(2)}$ in
that $\Pi$ should correspond with financial valuation $\pi^{(1)}$ for pure
financial claims and with the actuarial valuation when considering pure
actuarial claims.

\begin{definition}
[Orthogonal-consistency]A valuation $\Pi:L^{\infty}(\mathcal{F})\rightarrow
\mathbb{R}$ is said to be orthogonal consistent with the financial valuation
$\pi^{(1)}$ and the actuarial valuation $\pi^{(2)}$ if we have that for
$i=1,2$,
\begin{equation}
\Pi \left[  S^{\left(  i,\perp \right)  }\right]  =\pi^{(i)}\left[  S^{\left(
i,\perp \right)  }\right]  ,\  \text{if}\ S^{\left(  i,\perp \right)  }%
\in L^{\infty}(\mathcal{F}^{\left(  i,\perp \right)  }).\label{AFconsistent}%
\end{equation}

\end{definition}

We will show later that the two-step actuarial valuation introduced in this
paper is orthogonal consistent with the financial and actuarial valuations.

Similar to \cite{dhaene2017}, a hybrid claim is a claim which depends on both
the actuarial as well as the financial information, i.e.\ a hybrid claim $S$
can be expressed as follows:
\[
S\  \text{is a hybrid claim}\leftrightarrow S\in L^{\infty}(\mathcal{F}).
\]
Different valuation frameworks can be considered depending on how financial
and actuarial valuations are merged together. In Section \ref{ACvaluations},
we propose a two-step valuation which applies a financial valuation after
conditioning on actuarial information. For this reason, we briefly introduce
the concept of conditional valuations hereafter.

\section{Actuarial-consistent valuations}\label{ACvaluations}

Given a financial valuation principle $\pi^{(1)}$ and an actuarial valuation principle $\pi^{(2)}$, we search for general valuations $\Pi$ that are consistent with both the financial valuation $\pi^{(1)}$ and the actuarial valuation $\pi^{(2)}$, and study their properties. 

This section starts with introducing the concept of \textit{actuarial-consistency}. Condition \eqref{AFconsistent} for a valuation $\Pi$ states that actuarial claims which are independent of the financial information, should be valued using the actuarial valuation principle $\pi^{(2)}$. A valuation $\Pi$ is actuarial-consistent if all actuarial claims are priced with an actuarial valuation, even the ones that may be dependent to financial information. 

\begin{definition}
[Actuarial-consistency] A valuation $\Pi$ is called
actuarial consistent (ACV) with an actuarial valuation $\pi^{(2)}$ if for any actuarial claim $S^{(2)}\in
L^{\infty}(\mathcal{F}^{(2)})$ the following holds:
\begin{equation}
\Pi \left[  S^{(2)}\right]  =\pi^{(2)}\left[  S^{(2)}\right]
.\label{Def:WeakACV}%
\end{equation}

\end{definition}%

Actuarial consistency postulates that an actuarial valuation is applied for all actuarial claims. Note that actuarial consistency is stronger than condition \eqref{AFconsistent}, which only states that independent actuarial claims are priced using the actuarial valuation. In \cite{dhaene2017}, the authors define a similar notion of actuarial consistency, but the condition only holds for the claims which are independent of the financial filtration $\mathbb{F}^{(1)}$. 

In this section we define two new classes of actuarial-consistent valuations. The first class are the actuarial hedge-based valuations and the second class are the two-step actuarial valuations.  

\subsection{Actuarial hedge-based valuations}
In this section we consider a one-period setting, where we value a claim at time $t=0$ and determine a corresponding risk management strategy. 

We start with introducing the class of actuarial-consistent hedgers. In \cite{dhaene2017}, the authors showed that any market-consistent valuation can be represented as the time-0 value of a market-consistent hedger in a one-period setting. This result was further generalized in \cite{barigoumultiperiod,chen2021fair}. Similarly, in this section, we establish the relationship between actuarial-consistent valuations and their corresponding hedgers.

A trading strategy $\nu$ is a real-valued vector $\left( \nu_0, \nu_1,\ldots, \nu_{n^{(1)}}\right)$ where the component $\nu_i$ denotes the number of units invested in asset $i$ at time $t=0$ until maturity with the asset 0 is the risk-free asset with constant interest rate $r$. We denote the set of all \bluebis{these static} trading strategies by $\Theta$. 
\begin{definition}
A hedger $\theta: L^{\infty}(\mathcal{F})\rightarrow \Theta $ is a function which maps a claim $S\in L^{\infty}(\mathcal{F})$ into a trading strategy $\theta_S$ and satisfies the following conditions
\begin{enumerate}
\item $\theta$ is normalized: $\theta_0=(0,0,\ldots,0)$.
\item $\theta$ is translation invariant: $\theta_{S+a}=\theta_S+\left(a, 0, \ldots, 0 \right)$, for any $S \in L^{\infty}(\mathcal{F})$ and $a\in\mathbb{R}.$
\end{enumerate}
\end{definition}
The trading strategy $\theta_S$ is called the hedge for the claim $S$. Now, we define the class of actuarial-consistent hedgers.

\begin{definition}
A hedger is said to be an actuarial-consistent hedger if there exists a valuation $\pi$ such that
$$\theta_{S^{(2)}} = \left(\pi[S^{(2)}], 0,0,\ldots,0 \right),\ \text{for any }\ S^{(2)}\in L^{\infty}(\mathcal{F}^{(2)}).  $$
\end{definition}
An actuarial-consistent hedger will only allow investments in the risk-free bank account for actuarial claims. The higher potential returns in risky assets can be used to protect against the future losses from a claim. However, when an investment in risky assets is used for managing an actuarial claim, the insurer will be exposed to movements on the financial market. By considering actuarial-consistent hedgers, the insurer only adds risky investments to the portfolio if the claim he is trying to hedge contains financial risks.

\begin{example}[Actuarial-consistent hedgers]
	Hereafter, we consider two examples of actuarial-consistent hedgers.
\begin{enumerate}
\item Consider the hedger $\theta$ such that
$$
\theta_S = \begin{cases}
 \left(\pi^{(2)}[S], 0,0,\ldots,0 \right),\ \text{for}\ S\in L^{\infty}(\mathcal{F}^{(2)})\\
\arg \min _{\boldsymbol{\mu} \in \Theta} \mathbb{E}^{\mathbb{P}}\left[(S-\boldsymbol{\mu} \cdot \boldsymbol{Y})^{2}\right],\ \text{for any}\ S\in L^{\infty}(\mathcal{F}) \setminus L^{\infty}(\mathcal{F}^{(2)}) ,
\end{cases}
$$
where $\boldsymbol{Y}$ is the financial risks defined in subsection
\ref{FV} and $\Theta$ is the set of all hedging strategies, that is the set of all $(n^{(1)}+1)$-
dimensional real-valued vectors. Such hedger invests risk-free for actuarial claims and invests following quadratic hedging for all remaining claims. By construction, such hedger is actuarial consistent. \bluebis{For details on risk-minimizing strategies for insurance processes, we refer to \citep{moller2001risk,delong2019fair}.}
\item Assume that an insurer considers the two-step actuarial valuation for any claim $S$:
	$$\Pi\left[S \right]=\pi^{(2)}\left[ \pi\left[S|\  \mathcal{F}^{(2)} \right]\right],$$
	where $\pi \left[  \cdot|\mathcal{F}^{(2)}\right]  :L^{\infty}%
(\mathcal{F})\longrightarrow L^{\infty}(\mathcal{F}^{(2)})$ is an $\mathcal{F}^{(2)}$-conditional
valuation and $\pi^{(2)}$ is an actuarial valuation.

	If the insurer invests the whole value in risk-free asset, the following hedger is used: 
	\begin{equation}\label{H1}
	\theta_{S} = \left(\Pi\left[S \right], 0,0,\ldots,0 \right),\ \text{for any}\ S \in L^{\infty}(\mathcal{F}). 
	\end{equation}
	Such hedger is naturally an actuarial-consistent hedger. 
	$\hfill\blacktriangleleft$
\end{enumerate}
\end{example}

In the following lemma, we show how one can decompose a hybrid claim into an actuarial and financial part and define an actuarial-consistent hedger.
\begin{lemma}\label{Lemma3}
Consider a hybrid claim $S \in L^{\infty}(\mathcal{F})$ and a hedger $\tilde{\theta}$. We decompose the claim $S$ into two parts as follows: 
\begin{eqnarray*}
H_S^{(2)} & = & \mathbb{E}\left[S|\ \mathcal{F}^{(2)} \right] - \mathbb{E}\left[S \right], \\
H_S^{(1)} & = & S- H_S^{(2)}.
\end{eqnarray*}
Note that $H_S^{(2)} \in L^{\infty}(\mathcal{F}^{(2)})$.
Define the hedger $\theta_S$ as follows:
$$\theta_S=\tilde{\theta}_{H_S^{(1)}}.$$
Then, $\theta_S$ is an actuarial hedger. 
\end{lemma}
\begin{proof}
Consider an actuarial claim $S^{(2)}\in L^{\infty}(\mathcal{F}^{(2)})$. Then $S^{(2)}$ is $\mathcal{F}^{(2)}-$measurable and therefore $H_{S^{(2)}}^{(2)}=S^{(2)} - \mathbb{E}\left[S^{(2)} \right]$  and $H_{S^{(2)}}^{(1)}=\mathbb{E}\left[S^{(2)} \right] \in \mathbb{R}$. Since a hedger $\tilde{\theta}$ is translation invariant, we then find:
$$\theta_{S^{(2)}}=\tilde{\theta}_{H_{S^{(2)}}^{(1)}}=\left(\mathbb{E}\left[S^{(2)}\right], 0,0,\ldots,0 \right) ,$$
which shows we have an actuarial-consistent hedger.
\end{proof}

Lemma \ref{Lemma3} illustrates how one can derive an actuarial-consistent hedger using a general hedger. The claim $H_S^{(2)}$ can be interpreted as the actuarial part of the claim $S$. 

Assume that we have a claim $S$ and a hedger $\theta$. If we want to invest in the hedge $\theta_S$, we have to pay its time-$0$ value which is given by $\pi^{(1)}\left[ \theta_S\cdot \boldsymbol{Y}\right]$. In the next lemma, we show that the resulting financial value is actuarial consistent if the hedger is actuarial consistent. 

\begin{lemma}\label{Lemma_3} The following two statements are equivalent.
\begin{enumerate}
\item The valuation $\Pi$ can be expressed as follows 
\begin{equation}\label{L3.3}
    \Pi[S]=\pi^{(1)}\left[\theta_S\cdot \boldsymbol{Y}\right],
\end{equation}
    where $\theta$ is an actuarial-consistent hedger.
    \item $\Pi$ is actuarial consistent.
\end{enumerate}

\end{lemma}

\begin{proof}
Assume $\theta$ is  an actuarial-consistent hedger and the valuation $\Pi$ is given by \eqref{L3.3}. Then for any actuarial claim $S^{(2)} \in L^{\infty}(\mathcal{F}^{(2)})$, we have that $\theta_{S^{(2)}}\cdot \boldsymbol{Y}= \pi^{(2)}\left[S^{(2)}Y_0\right]$. Taking into account $Y_0=1$, we find that $\Pi$ is  actuarial consistent.
Assume now that $\Pi$ is an actuarial-consistent valuation. Defining $\theta$ as in \eqref{H1} shows that $\Pi[S]=\pi^{(1)}\left[\theta_S\cdot \boldsymbol{Y}\right], $ where $\theta$ is an actuarial-consistent hedger.
\end{proof}

In order to determine a hedge-based value of $S,$ one first splits this claim into a hedgeable claim, which (partially) replicates $S$, and a remaining claim. The value of the claim $S$ is then defined as the sum of the financial value of the hedgeable claim and the actuarial value of the remaining claim, determined according to a pre-specified actuarial valuation.

\begin{definition}
A valuation $\Pi$ is called an actuarial hedge-based valuation with financial valuation $\pi^{(1)}$ and \bluebis{actuarial-consistent valuation} $\pi$ if it can be expressed as follows:
$$\Pi[S]=\pi^{(1)}\left[ \theta_S\cdot \boldsymbol{Y}\right] + \pi\left[S- \theta_S\cdot \boldsymbol{Y}\right],$$
where $\theta_S$ is an actuarial-consistent hedger.
\end{definition}

In the following theorem, we prove that the class of
 actuarial-consistent valuations is equal to the class of actuarial hedge-based valuations.

\begin{theorem}\label{thactuarialHB}
    A valuation $\Pi$ is an  actuarial-consistent valuation if, and only if, it is an actuarial hedge-based valuation.
\end{theorem}
\begin{proof}
Assume $\Pi$ is an actuarial hedge-based valuation, i.e.\ we have that
$$\Pi[S]=\pi^{(1)}\left[\theta_S\cdot \boldsymbol{Y}\right]+\pi \left[S-\theta_S\cdot \boldsymbol{Y}\right],$$
where $\theta_S$ is an actuarial hedger. Consider an actuarial claim $S^{(2)}\in L^{\infty}(\mathcal{F}^{(2)})$. Then $$\theta_{S^{(2)}}=\left(\pi^{(2)}\left[S^{(2)} \right],0,0,\ldots,0 \right),$$ for some actuarial valuation $\pi^{(2)}$. 
Then it is straightforward to verify that
$$\Pi\left[S^{(2)} \right]=\pi \left[S^{(2)}\right],$$
which shows that an actuarial hedge-based valuation is a  actuarial-consistent valuation.

Assume now that $\Pi$ is an actuarial-consistent valuation. Then it follows from Lemma \ref{Lemma_3} that 
$$\Pi[S]=\pi^{(1)}\left[\theta_S\cdot \boldsymbol{Y} \right],$$
where 
$\theta_S=\left(\Pi[S],0,\ldots,0 \right)$.
Define the valuation $\Pi'$ as follows:
$$\Pi'[S]= \pi^{(1)}\left[\theta_S\cdot \boldsymbol{Y} \right]+\Pi\left[S-\theta_S\cdot \boldsymbol{Y} \right].$$
Then $\Pi'$ is an actuarial hedge-based valuation. Moreover, since $\theta_S\cdot \boldsymbol{Y}=\Pi[S]$, we also find that $\Pi'=\Pi$, which ends the proof.
\end{proof}

We remark that all the results in this subsection are satisfied for general valuations that are only normalized and translation-invariant. Positive homogeneity and convexity properties, which are assumed for a valuation in Definition \ref{def}, are not necessary for the equivalence result of Theorem \eqref{thactuarialHB}  to hold.

\subsection{Two-step actuarial valuations}

Hereafter, we introduce a class of actuarial-consistent valuations which we
call \textit{two-step actuarial valuations}. More specifically, in a first
step we compute the financial value of $S$ conditional on actuarial scenarios
(the values of the actuarial assets $X$), i.e. $\pi \left[  S|\  \mathcal{F}%
^{(2)}\right]  $. Since this conditional payoff depends only on actuarial
scenarios and is then $\mathcal{F}^{(2)}$-measurable, in the second step the
quantity $\pi \left[  S|\  \mathcal{F}^{(2)}\right]  $ should be valuated via a
standard actuarial valuation $\pi^{(2)}$. \bluebis{We remark that the two-step actuarial valuation is the reversed version of the two-step market valuation of \cite{pelsser2014time}, where financial and actuarial valuations are applied in different orders. We will show in Lemma \ref{lemmaind}  that under strong assumptions, both valuations lead to the same result.}
\blue
\begin{definition}
[Two-step actuarial valuation] The valuation $\Pi$ is called a two-step
actuarial valuation if there exists an $\mathcal{F}^{(2)}$-conditional
valuation $\pi \left[  \cdot|\mathcal{F}^{(2)}\right]  :L^{\infty}%
(\mathcal{F})\longrightarrow L^{\infty}(\mathcal{F}^{(2)})$ such that
\begin{equation}\label{eq:twostepactuarial}
    \Pi \left[  S\right]  =\pi^{(2)}\left[  \pi \left[  S|\  \mathcal{F}%
^{(2)}\right]  \right]  ,
\end{equation}
where $\pi^{(2)}:L^{\infty}(\mathcal{F}^{(2)})\rightarrow \mathbb{R}$ is an
actuarial valuation. 
\end{definition}
\black

Hence, the two-step actuarial valuation consists of applying the
market-adjusted valuation to the residual risk which remains after having
conditioned on the future development of the actuarial risks, i.e.
$\mathcal{F}^{(2)}$. It is straightforward to verify that the two-step
actuarial valuation is orthogonal consistent.

\blue
\begin{example}
As an example of a two-step actuarial valuation, we note that \cite{moller2002valuation} proposed a modified version of the standard deviation principle: 
    \begin{equation}
\Pi\left[S \right]=\mathrm{E}^{\mathbb{Q}}\left[S \right]+a\left(\operatorname{Var}\left[\mathrm{E}^{\mathbb{Q}}\left[S \mid \mathcal{F}^{(2)}\right]\right]\right)^{1 / 2},
\end{equation}
which corresponds to applying the traditional standard deviation principle to the no-arbitrage price of $S$ conditional on the actuarial filtration (see Equation (5.5) in \cite{moller2002valuation}).
\end{example}\black

In the following theorem, we show that any two-step actuarial valuation is actuarial-consistent. Moreover, if the valuation is coherent, we provide a characterization of the two-step actuarial valuation. 

\begin{theorem}
[Characterization of actuarial consistency]\label{characterization} The
following  properties hold:

\begin{enumerate}

\item Any two-step actuarial valuation $\Pi$ is  actuarial-consistent.

\item \blue If $\Pi$ is a coherent and actuarial-consistent valuation with a linear
actuarial valuation $\pi^{(2)}:L^{\infty}(\mathcal{F}^{(2)})\rightarrow
\mathbb{R}$, then there exists an $\mathcal{F}^{(2)}$-conditional coherent
valuation $\pi \left[  \cdot|\mathcal{F}^{(2)}\right]  :L^{\infty}%
(\mathcal{F})\rightarrow L^{\infty}(\mathcal{F}^{(2)})$ such that for any
$S\in L^{\infty}(\mathcal{F})$,
\[
\Pi \lbrack S]=\mathbb{\pi}^{(2)}\left[  \pi \left[  S|\mathcal{F}^{(2)}\right]
\right]  ,
\]
where
\[
\pi \lbrack S|\mathcal{F}^{(2)}]=\underset{Z\in \mathcal{R}}{\mathrm{ess\,sup}%
}\, \mathbb{E}[ZS|\mathcal{F}^{(2)}],
\]
with $\mathcal{R}=\{ \xi \in L^{1}(\mathcal{F}):\mathbb{E}[\xi|\mathcal{F}%
^{(2)}]=1\}$.\black
\end{enumerate}
\end{theorem}

\begin{proof}
(1) To prove that $\Pi$ is  actuarial consistent, it is sufficient to
notice that for any $S^{(2)}\in L^{\infty}(\mathcal{F}^{(2)})$,
\begin{align*}
\Pi \left[  S^{(2)}\right]   &  =\pi^{(2)}\left[  \pi \left[  S^{(2)}%
|\mathcal{F}^{(2)}\right]  \right]  \\
&  =\pi^{(2)}\left[  S^{(2)}\right]  ,
\end{align*}
where we have used that $S^{(2)}$ is $\mathcal{F}^{(2)}$-measurable.

(2) \blue Because $\Pi$ is coherent, we have that
\[
\Pi \lbrack S]=\sup_{\xi \in \mathcal{M}}\mathbb{E}\left[  \xi S\right]  .
\]
Since $\pi^{(2)}$ is linear, there exists $\varphi^{(2)}\in \mathcal{M}%
^{(2)}\mathcal{=}\left \{  \xi \in L^{1}(\mathcal{F}^{(2)}):\mathbb{E}%
[\xi]=1\right \}  $ such that
\[
\pi^{(2)}\left[  S^{(2)}\right]  =\mathbb{E}\left[\varphi^{(2)}S^{(2)}\right].
\]
By the same arguments in the proof of Proposition 3.3 of
\cite{pelsser2014time}, we can decompose $\xi$ as $\xi=\varphi^{(2)}Z$ with
$Z\in \mathcal{R}=\{ \xi \in L^{1}(\mathcal{F}):\mathbb{E}[\xi|\mathcal{F}%
^{(2)}]=1\}$. Thus
\begin{align*}
\Pi \lbrack S] &  =\sup_{\xi \in \mathcal{M}}\mathbb{E}\left[  \xi S\right]  \\
&  =\sup_{Z\in \mathcal{R}}\mathbb{E}\left[  \varphi^{(2)}ZS\right]  \\
&  =\sup_{Z\in \mathcal{R}}\mathbb{E}\left[  \varphi^{(2)}\mathbb{E}%
[ZS|\mathcal{F}^{(2)}]\right]  .
\end{align*}
Now we define $\pi \lbrack S|\mathcal{F}^{(2)}]=\underset{Z\in \mathcal{R}}{\mathrm{ess\,sup}%
}\,\mathbb{E}[ZS|\mathcal{F}^{(2)}]$. It is straightforward to verify
that $\pi \lbrack S|\mathcal{F}^{(2)}]$ is a $\mathcal{F}^{(2)}$-conditional
valuation. All we are left to show is
\[
\sup_{Z\in \mathcal{R}}\mathbb{E}\left[  \varphi^{(2)}\mathbb{E}[ZS|\mathcal{F}%
^{(2)}]\right]  =\mathbb{E}\left[  \varphi^{(2)}\underset{Z\in \mathcal{R}%
}{\mathrm{ess\,sup}%
}\,\mathbb{E}[ZS|\mathcal{F}^{(2)}]\right]  .
\]
This can be proved using the same arguments in the proof of Theorem 3.10 of
\cite{pelsser2014time}.\black
\end{proof}
\bigskip

Combining Theorems \ref{thactuarialHB} and \ref{characterization}, we find the following representation corollary.

\begin{corollary}
The following holds:
\begin{enumerate}
\item Any two-step actuarial and any actuarial hedge-based valuation is actuarial-consistent. 
\item Any actuarial-consistent valuation is an actuarial hedge-based valuation. 
\end{enumerate}
\end{corollary}

\subsection{ Comparisons of two-step actuarial valuation and two-step financial valuation}\label{sectionFV}

Motivated by solvency regulations, the recent actuarial literature focused on market-consistent valuations which essentially require that financial risks are priced with a risk-neutral valuation. In this section, we provide a detailed comparison between actuarial-consistent and market-consistent valuations. 



\begin{definition}[Strong market-consistency]
	A valuation $\Pi$ is called strong market-consistent (strong MCV) if for any financial claim $S^{(1)}$ the following holds:
	\begin{equation}\label{Def:StrongMCV}
	\Pi\left[S+ S^{(1)}\right]=\Pi\left[ S\right]+\pi^{(1)}\left[ S^{(1)}\right].
	\end{equation}
\end{definition}
In the literature, market-consistency is usually defined via a condition identical or similar to the condition \eqref{Def:StrongMCV} (see e.g. \cite{pelsser2014time}, \cite{dhaene2017} and \cite{barigou2018fair}). \blue When the financial valuation $\pi^{(1)}$ is linear, Proposition 3.3 of \cite{pelsser2014time} shows that the
strong market-consistency is equivalent to the following weak market-consistency. \black

\begin{definition}[Weak market-consistency]
	A valuation $\Pi$ is called weak market-consistent (weak MCV) if for any financial claim $S^{(1)}$ the following holds:
	\begin{equation}\label{Def:WeakMCV}
	\Pi\left[ S^{(1)}\right]=\pi^{(1)}\left[ S^{(1)}\right].
	\end{equation}
\end{definition}

We remark that \cite{assa2018market} also investigated these two types of market-consistency, that they called market-consistency of type I and type II.

Following the definition of the two-step market valuation in \cite{pelsser2014time}, we define the class of two-step financial valuations.
\begin{definition}
[Two-step financial valuation]\label{Def:2-finstep} The valuation $\Pi$ is
called a two-step financial valuation if there exists an $\mathcal{F}^{(1)}%
$-conditional valuation $\pi \left[  \cdot|\mathcal{F}^{(1)}\right]
:L^{\infty}(\mathcal{F})\longrightarrow L^{\infty}(\mathcal{F}^{(1)})$ such
that
\[
\Pi \left[  S\right]  =\pi^{(1)}\left[  \pi \left[  S|\  \mathcal{F}%
^{(1)}\right]  \right]  ,
\]
where $\pi^{(1)}:L^{\infty}(\mathcal{F}^{(1)})\rightarrow \mathbb{R}$ is a
financial valuation.
\end{definition}

After having defined two broad classes of two-step valuations: market-consistent and actuarial-consistent valuations, a natural question arises: Could we always define a \textit{fair} valuation that is market-consistent \textit{and} actuarial-consistent? 
\begin{definition}[Fair valuation]
	A two-step valuation $\Pi$ is fair if it is weak market consistent and actuarial consistent. 
\end{definition}

In general, it will not always be possible to define a fair valuation. Indeed, in a general probability space in which financial and actuarial risks are dependent, there is ambiguity on the valuation to be used: a market-consistent valuation calibrated to market prices \textit{or} an actuarial-consistent valuation calibrated to historical actuarial data.

In the following lemma, we show that when financial risks and actuarial risks
are independent, then the two-step financial valuation and the two-step actuarial valuation coincide. 
\blue
\begin{lemma}\label{lemmaind}
Assume that $\mathcal{F}^{(1)}$ and $\mathcal{F}^{(2)}$ are independent and
for any $S\in L^{\infty}(\mathcal{F})$ it can be expressed as $S=S^{(1)}%
S^{(2)}$ with $S^{(1)}\in L^{\infty}(\mathcal{F}^{(1)})$ and $S^{(2)}\in
L^{\infty}(\mathcal{F}^{(2)})$. If a two-step valuation $\Pi:L^{\infty
}(\mathcal{F})\longrightarrow \mathbb{R}$ with a linear second step valuation
is a fair valuation, then
\[
\Pi \lbrack S]=\pi^{(1)}\left[  S^{(1)}\right]  \pi^{(2)}\left[S^{(2)}\right],
\]
where $\pi^{(1)}:L^{\infty}(\mathcal{F}^{(1)})\longrightarrow \mathbb{R}$ and
$\pi^{(2)}:L^{\infty}(\mathcal{F}^{(2)})\longrightarrow \mathbb{R}$.
\end{lemma}

\begin{proof}
Since $\Pi$ is actuarial consistent with a linear second step actuarial valuation, there
exists a $\mathcal{F}^{(2)}$-conditional valuation $\pi \left[  \cdot
|\mathcal{F}^{(2)}\right]  :L^{\infty}(\mathcal{F})\longrightarrow L^{\infty
}(\mathcal{F}^{(2)})$ such that
\begin{align*}
\Pi \left[  S\right]   &  =\pi^{(2)}\left[  \pi \left[  S^{(1)}S^{(2)}%
|\mathcal{F}^{(2)}\right]  \right]  \\
&  =\pi^{(2)}\left[  S^{(2)}\pi \left[  S^{(1)}|\mathcal{F}^{(2)}\right]
\right]  \\
&  =\pi^{(2)}\left[S^{(2)}\right]\Pi \left[  S^{(1)}\right]  ,
\end{align*}
where in the second step we used the positive homogeneity of $\pi \left[
\cdot|\mathcal{F}^{(2)}\right]  $ and the last step is due to the independence
of $\mathcal{F}^{(1)}$ and $\mathcal{F}^{(2)}$. Because $\Pi$ is also weak
market consistent, we have
\[
\Pi \left[  S^{(1)}\right]  =\pi^{(1)}\left[  S^{(1)}\right]  .
\]
Then the desired result follows.
\end{proof}
\black

With the emergence of the market for longevity derivatives, a valuator needs to make a choice between market-consistency and actuarial-consistency. For instance, consider a market with some traded longevity bonds and there is an issue of a new longevity product. One needs to decide to use either a market-consistent approach based on the traded longevity bonds in the market or an actuarial-consistent approach based on longevity trend assumptions. 

In the following example, we illustrate this point and compare a market-consistent and an actuarial-consistent valuation in the presence of a longevity bond. 
In particular, we compare the following two valuations: for any $S \in L^{\infty}(\mathcal{F})$, we define
\begin{align}
\Pi^{(1)}[S] &  =\mathbb{\pi}^{(2)}\left[  \mathbb{E}^{\mathbb{Q}}\left[
S|\mathcal{F}^{(2)}\right]  \right]  ,\label{p1}\\
\Pi^{(2)}[S] &  =\mathbb{E}^{\mathbb{Q}}\left[  \mathbb{\pi}^{(2)}\left[
S|\mathcal{F}^{(1)}\right]  \right]  .\label{p2}%
\end{align}
Here for simplicity, we abuse the notation a little. In (\ref{p1}), the
conditional valuation $\mathbb{E}^{\mathbb{Q}}\left[  S|\mathcal{F}%
^{(2)}\right]  $ is in fact $\mathbb{E}^{\mathbb{P}}\left[  ZS|\mathcal{F}%
^{(2)}\right]  $ for some $Z\in \{ \xi \in L_{+}^{1}(\mathcal{F}):\mathbb{E}%
[\xi|\mathcal{F}^{(2)}]=1\}$. That is $\mathbb{Q}$ in (\ref{p1}) is an
absolutely continuous probability measure with respect to $\mathbb{P}$
conditional on $\mathcal{F}^{(2)}$. In (\ref{p2}), $\mathbb{Q}$ is an
absolutely continuous probability measure with respect to $\mathbb{P}$ in the
usual sense; that is $\mathbb{E}^{\mathbb{Q}}\left[  S^{(1)}\right]
=\mathbb{E}^{\mathbb{P}}\left[  ZS^{(1)}\right]  $ for some $Z\in \{ \xi \in
L_{+}^{1}(\mathcal{F}^{(1)}):\mathbb{E}[\xi]=1\}$. $\mathbb{\pi}^{(2)}$ in
(\ref{p1}) is the usual actuarial valuation defined on $L^{\infty}%
(\mathcal{F}^{(2)})$ while $\mathbb{\pi}^{(2)}\left[  S|\mathcal{F}%
^{(1)}\right]  $ in (\ref{p2}) is a conditional valuation defined on
$L^{\infty}(\mathcal{F})$ sharing similar structure with $\mathbb{\pi}^{(2)}$.
For example, if $\mathbb{\pi}^{(2)}[S^{(2)}]=\mathbb{E}^{\mathbb{P}}[S^{(2)}%
]$, then $\mathbb{\pi}^{(2)}\left[  S|\mathcal{F}^{(1)}\right]  =\mathbb{E}%
^{\mathbb{P}}\left[  S|\mathcal{F}^{(1)}\right]  $.

\begin{example}[Comparison between MCV and ACV]\label{2stepexample4}
	(a) Consider a portfolio of pure endowments for $l_{x}$ Belgian insureds of age $x$ at time $0$. The pure endowment guarantees a sum of 1 if the policyholder is still alive at maturity. The aggregate payoff can be written as
	\begin{equation*}
	S=L_{x+T},
	\end{equation*}
	with $L_{x+T}$ the number of policyholders who survive up to the maturity time $T$. Moreover, we assume that the financial market is composed of two assets: a risk-free asset $Y^{(0)}(t)=e^{rt}$ and a longevity bond for which the payoff at maturity is $Y^{(1)}(T)=\tilde{L}_{x+T}$, the equivalent of $L_{x+T}$ but for the Dutch population. 
	First, we determine the value by the two-step actuarial valuation:
	\begin{align*}
	\Pi^{(1)}[S]   &  =\mathbb{E}^{\mathbb{P}}\left[  \mathbb{E}^{\mathbb{Q}}\left[ e^{-rT} L_{x+T}|\mathcal{F}%
	^{(2)}\right]  \right] \\
	&  = \mathbb{E}^{\mathbb{P}}\left[ e^{-rT} L_{x+T} \mathbb{E}^{\mathbb{Q}}\left[1  |\mathcal{F}^{(2)}\right]  \right] \\
	&  = e^{-rT} \mathbb{E}^{\mathbb{P}}\left[ L_{x+T}  \right].
	\end{align*}
	The actuarial-consistent valuation would suggest a full investment in the risk-free asset. Secondly, assuming that the Belgian population live slightly shorter than the Dutch population\footnote{For the reader interested in Dutch and Belgian mortality projections, we refer to \cite{antonio2017producing}}: $\mathbb{E}^{\mathbb{P}}\left[ L_{x+T}| \tilde{L}_{x+T} \right]=\beta\tilde{L}_{x+T} $ with $\beta<1$, we determine the value according to the two-step financial valuation:
	\begin{align*}
	\Pi^{(2)}[S]   &  =\mathbb{E}^{\mathbb{Q}}\left[e^{-rT}  \mathbb{E}^{\mathbb{P}}\left[L_{x+T}|\mathcal{F}%
	^{(1)}\right]  \right] \\
	&  =\mathbb{E}^{\mathbb{Q}}\left[e^{-rT} \beta \tilde{L}_{x+T} \right] \\
	&  = \beta Y^{(1)}(0),
	\end{align*}
	where $Y^{(1)}(0)$ is the current price of the longevity bond. The market-consistent valuation would then suggest a full investment in the Dutch longevity bond.
	
	(b) In order to better grasp the difference between the actuarial-consistent and market-consistent valuations, we introduce the following modelling assumptions. \newline
	Assume that the interest rate $r=0$, and the bivariate Belgian-Dutch population follows the distribution: $({L}_{x+T},\tilde{L}_{x+T}) \sim \mathcal{N}(\boldsymbol{\mu},\boldsymbol{\Sigma})$ with
	\begin{equation*}
	\boldsymbol\mu = \begin{pmatrix} \mu_1 \\ \mu_2 \end{pmatrix}, \quad
	\boldsymbol\Sigma = \begin{pmatrix} \sigma_1^2 & \rho \sigma_1 \sigma_2 \\
	\rho \sigma_1 \sigma_2  & \sigma_2^2 \end{pmatrix}.
	\end{equation*}
	Hence, both Belgian and Dutch populations are normal distributed with correlation $\rho$. Hereafter, we compare the two-step actuarial and financial valuations given by
	\begin{align}
	    	\Pi^{(1)}[S]  &  =\pi^{(2)}\left[  \mathbb{E}^{\mathbb{Q}}\left[S| \mathcal{F}^{(2)}\right]  \right] \label{twostep1}\\
	    	 	\Pi^{(2)}[S]  &  =\mathbb{E}^{\mathbb{Q}}\left[  \pi^{(2)}[S| \mathcal{F}^{(1)}]  \right] \label{twostep2}
	\end{align}
	where $\pi^{(2)}$ is the standard deviation principle \eqref{stdprinciple} and $\pi^{(2)}[\cdot| \mathcal{F}^{(1)}]$ is the conditional standard deviation principle:
	\begin{equation*}
	    \Pi [S| \mathcal{F}^{(1)} ]= \mathbb{E}^{\mathbb{P}}\left[S| \mathcal{F}^{(1)} \right]+\beta \text{Var}\left[S| \mathcal{F}^{(1)} \right].
	\end{equation*}
	The two-step actuarial valuation of $S=L_{x+T}$ is given by
	\begin{align}
	\Pi^{(1)}[S]  &  =\pi^{(2)}\left[  \mathbb{E}^{\mathbb{Q}}\left[ e^{-rT} L_{x+T}| L_{x+T}\right]  \right] \nonumber\\
	&  = \mathbb{E}^{\mathbb{P}}\left[ L_{x+T}  \right] + \beta \sigma^{\mathbb{P}}\left[ L_{x+T}  \right]  \nonumber\\
	&= \mu_1+ \beta \sigma_1.\label{be1ex4}
	\end{align}
	To determine the two-step financial valuation, we first notice that by standard results of normal distributions, we have
	\begin{equation*}
	L_{x+T}|\tilde{L}_{x+T}=x \sim \mathcal{N}\left(\mu_1+\rho\frac{\sigma_1}{\sigma_2}\left(x-\mu_2\right),  (1-\rho^2) \sigma_{1}^2\right).
	\end{equation*}
	Let us further assume that the distribution of $\tilde{L}_{x+T}$ under $\mathbb{Q}$ is
	\begin{equation*}
	\tilde{L}_{x+T}\overset{\mathbb{Q}}{\sim}\mathcal{N}(\mu_2-\sigma_2\kappa,\sigma_2^2),
	\end{equation*}	
	where $\kappa>0$ is the market price of risk for the longevity bond. Therefore, we find that
	\begin{align}
	\Pi^{(2)}[S]&=\mathbb{E}^{\mathbb{Q}}\left[\pi^{(2)} \left[L_{x+T}|\tilde{L}_{x+T}\right] \right]\nonumber\\
		&=\mathbb{E}^{\mathbb{Q}}\left[ \mu_1+\rho\frac{\sigma_1}{\sigma_2}\left(\tilde{L}_{x+T}-\mu_2\right) +\beta \sigma_{1} \sqrt{1-\rho^2}      \right]\nonumber\\
		&=\mu_1+\rho\frac{\sigma_1}{\sigma_2}\left(\mathbb{E}^{\mathbb{Q}}\left[\tilde{L}_{x+T}\right]-\mu_2 \right) +\beta \sigma_{1} \sqrt{1-\rho^2} \nonumber\\
		&=\mu_1-\rho\sigma_1 \kappa +\beta \sigma_{1} \sqrt{1-\rho^2} .\label{be2ex4}
	\end{align}
	We can compare the two-step actuarial valuation \eqref{be1ex4} with the two-step financial valuation \eqref{be2ex4}. Intuitively, the difference should reflect two aspects:
	\begin{enumerate}
		\item The dependence between Belgian and Dutch populations.
		\item The risk premium on the Dutch longevity bond.
	\end{enumerate} 
	 The results confirm the intuition: the difference is given by
	\begin{equation}
	\Pi^{(2)}[S]-\Pi^{(1)}[S]=\sigma_{1} \left[\beta \left( \sqrt{1-\rho^2}-1\right) - \rho \kappa \right].\label{diffequation}
	\end{equation}
	We observe that the higher the correlation $\rho$, the higher the difference (this reflects the point 1.). Moreover, the absolute difference is an increasing function of the market price of risk $\kappa$ (this reflects the point 2.).
	
If the valuator can choose between the risk-free investment or the Dutch longevity bond, he will go for the longevity bond if the benefits are higher than the costs, i.e. if the risk reduction of investing in the longevity bond is higher than the extra price he has to pay (given by Equation \eqref{diffequation}). The prices at time $0$ of both approaches and the residual losses at maturity are given in the table below: 
	\begin{center}
	$	\begin{array}{ l|l } 
	\text{\underline{Price at time $0$}} & \text{\underline{Residual loss at maturity}} \\ 
	\Pi^{(1)}[S]=\mu_1+ \beta \sigma_1 & R_1={L}_{x+T}-\mu_1- \beta \sigma_1 \sim \mathcal{N}(-\beta \sigma_1,\sigma_1^2)\\ 
	\Pi^{(2)}[S]=\mu_1-\rho\sigma_1 \kappa +\beta \sigma_{1} \sqrt{1-\rho^2} & R_2={L}_{x+T}-\left(\mu_1+\rho\frac{\sigma_1}{\sigma_2}\left(\tilde{L}_{x+T}-\mu_2\right) +\beta \sigma_{1} \sqrt{1-\rho^2} \right) \\
	& \sim \mathcal{N}(-\beta \sigma_{1} \sqrt{1-\rho^2},(1-\rho^2)\sigma_1^2)\\ 
	\hline
	\end{array}$
\end{center}
From the table, we observe that the investment in the longevity bond leads to a decrease in the volatility of the residual loss but an increase in the expected loss. Notice that in case of comonotonic or countermonotonic risks (i.e. $\rho=\pm1$), the claim $S$ can be completely hedged with the longevity bond and the residual loss $R_2$ equals 0. The valuator will typically go for the longevity bond if the risk reduction (computed in terms of Value-at-Risk for simplicity) is higher than the extra price to pay:
\begin{align*}
    \sigma_{1} \left[\beta \left( \sqrt{1-\rho^2}-1\right) - \rho \kappa \right]&<VaR_p[R_2]-VaR_p[R_1]\\
    & =\sigma_1 \left[\Phi^{-1} (p)-\beta  \right] (\sqrt{1-\rho^2}-1).
\end{align*}
On the other hand, if the longevity bond price is too high in comparison with the risk reduction, an actuarial-consistent valuation is then preferable. 	 $\hfill\blacktriangleleft$
\end{example}

\begin{remark}
	In this paper, we do not  argue that one method is better than another; each one has pros and cons. While the second method allows to transfer the risk to the financial market, it comes also with a price: the liabilities become totally dependent on the longevity bond. In particular, an adverse shock on the Dutch population or a counter-party's default will have a direct effect on the assets backing the liabilities. 
	
	More generally, as pointed out by \cite{vedani2017market}, market-consistent valuations are directly subject to market movements, and can lead to excess volatility, depending on the calibration sets chosen by the actuary. We also refer to \cite{rae2018review} for different concerns around the appropriateness of market-consistency to the insurance business. 
\end{remark}

In the next example, we consider the valuation of a hybrid claim via a two-step financial and actuarial valuation, and investigate the difference between the two-step operators. 

\begin{example}[Two-step valuations for hybrid claims]\label{example5TS}
	Consider an equity-linked contract for a life $(x)$, which pays the call option $(Y-K)_+$ in case the policyholder is alive at time $T=1$ and 0 otherwise. We recall that the stock $Y$ can go up to 200 or down to 50, the strike $K=100$ and the policyholder survival is modelled by the indicator $I$. Therefore, the payoff of this contract is given by 
	\begin{equation}\label{payoffS}
	S=(Y-K)_+ \times I =\left\{\begin{array}{cc}
	100, & \text{if }  Y=200, I=1,\\
	0, & \text{otherwise. }
	\end{array}\right.
	\end{equation}	
	Similar to Example \ref{2stepexample4}, we consider the two-step financial and actuarial valuations given by Equations \eqref{twostep1} and \eqref{twostep2} for the hybrid payoff \eqref{payoffS}: 
	\begin{enumerate}
		\item  \underline{Two-step financial valuation:} The value of $S$ is given by 
		\begin{align*}
		\Pi^{(1)}[S]&=\mathbb{E}^\mathbb{Q}\left[\mathbb{E}^{\mathbb{P}}\left[ S|\mathcal{F}^{(1)}\right]+\beta\sqrt{\text{Var}^{\mathbb{P}}\left[S|\mathcal{F}^{(1)}\right]} \right]\\
		&=\mathbb{E}^\mathbb{Q} \left[(Y-K)_+ \left(\mathbb{E}^\mathbb{P}\left[I|Y\right]+\beta \sigma^\mathbb{P} \left[I|Y\right] \right)\right].
		\end{align*}
		If we note that 
		\begin{equation*}
		\mathbb{E}^\mathbb{P}\left[I|Y\right]=\left\{\begin{array}{cc}
		\mathbb{P}[I=1|Y=50], & \text{if }  Y=50,\\
		\mathbb{P}[I=1|Y=200], & \text{if }  Y=200,
		\end{array}\right.
		\end{equation*}
	then we find that the two-step financial value of $S$ is
	\begin{equation}\label{twostepfin}
	\Pi^{(1)}[S]=100\, q_Y\left(p_{I|Y=200}+\beta \sqrt{p_{I|Y=200} (1-p_{I|Y=200})} \right),
	\end{equation}
	where $q_Y$ is the $\mathbb{Q}$-probability that $Y$ goes up: $q_Y=\mathbb{Q}[Y=200]$ and $p_{I|Y=200}$ is the $\mathbb{P}$-probability that the policyholder is alive given that the stock goes up: $p_{I|Y=200}=\mathbb{P}[I=1|Y=200]$.
	
	\item \underline{Two-step actuarial valuation:} The value of $S$ is given by 
	\begin{align*}
	\Pi^{(2)}\left[S\right]&=\mathbb{E}^\mathbb{P}\left[\mathbb{E}^{\mathbb{Q}}\left[ S|\mathcal{F}^{(2)} \right]\right]+\beta \sigma^\mathbb{P}\left[\mathbb{E}^{\mathbb{Q}}\left[ S|\mathcal{F}^{(2)}\right]\right]\\
	&= \mathbb{E}^\mathbb{P} \left[I \mathbb{E}^{\mathbb{Q}}\left[ (Y-K)_+|I \right] \right]+\beta \sigma^\mathbb{P} \left[I \mathbb{E}^{\mathbb{Q}}\left[ (Y-K)_+|I \right] \right].
	\end{align*}
	Noting that 
	\begin{equation*}
	\mathbb{E}^\mathbb{Q}\left[(Y-K)_+|I\right]=\left\{\begin{array}{cc}
	100 \ \mathbb{Q}[Y=200|I=0], & \text{if }  I=0,\\
	100 \ \mathbb{Q}[Y=200|I=1], & \text{if }  I=1,
	\end{array}\right.
	\end{equation*}
	then we find that the two-step actuarial value of $S$ is
	\begin{equation}\label{twostepact}
	\Pi^{(2)}\left[S\right]=100\,q_{Y|I=1}\left(p_I+\beta \sqrt{p_I (1-p_I)} \right),
	\end{equation}
	where $p_I$ is the $\mathbb{P}$-probability that the policyholder is alive: $p_I=\mathbb{P}[I=1]$ and $q_{Y|I=1}$ is the $\mathbb{Q}$-probability that the stock goes up given that the policyholder is alive: $q_{Y|I=1}=\mathbb{Q}[Y=200|I=1]$.
	\end{enumerate}
	
If we compare the two-step financial and actuarial values \eqref{twostepfin} and \eqref{twostepact}, the structure is similar but the dependence between financial and actuarial risks is taken into account differently. In the first case, it is via the $\mathbb{P}$-probability of actuarial risks given financial scenarios, i.e. $p_{I|Y=200}$ while in the second case, it is via the $\mathbb{Q}$-probability of financial risks given actuarial scenarios, i.e. $q_{Y|I=1}$. In case of independence under $\mathbb{P}$ and $\mathbb{Q}$, both valuations are equal. In case of dependence, both valuations \eqref{twostepfin} and \eqref{twostepact} will in general be different as we illustrate below: 
\begin{align*}
\Pi^{(2)}\left[S\right]-\Pi^{(1)}\left[S\right]&=100\,q_{Y|I=1}\,p_I-100\,q_Y\,p_{I|Y=200}\\
&+100\, q_{Y|I=1}\,\beta \sqrt{p_I (1-p_I)}-100\, q_Y\,\beta \sqrt{p_{I|Y=200} (1-p_{I|Y=200})}.
\end{align*}
Let us further assume that the difference between $\mathbb{P}$ and $\mathbb{Q}$ is given by a constant market price of risk $\kappa$: 
\begin{align*}
	\kappa&=q_Y-p_Y,\\
	&= q_{Y|I=1}- p_{Y|I=1}.
\end{align*}
Therefore, by Bayes' Theorem, we find that
\begin{align*}
\Pi^{(2)}\left[S\right]-\Pi^{(1)}\left[S\right]&=100\left(\frac{p_{I|Y=200}p_Y}{p_I}+\kappa\right)p_I-100\left(p_Y+\kappa\right)p_{I|Y=200}\\
&+100\, (p_{Y|I=1}+\kappa)\,\beta \sqrt{p_I (1-p_I)}\\
&-100\, (p_Y+\kappa)\,\beta \sqrt{p_{I|Y=200} (1-p_{I|Y=200})}.
\end{align*}
After simplifications, we find that 
\begin{align*}
\Pi^{(2)}\left[S\right]-\Pi^{(1)}\left[S\right]&=100\kappa\left({p_I}-p_{I|Y=200}\right)\\
&+100\kappa\beta\left(\sqrt{p_I (1-p_I)}-\sqrt{p_{I|Y=200} (1-p_{I|Y=200})}\right)\\
&+100 p_{Y|I=1}\beta\sqrt{p_I (1-p_I)}-100p_Y\beta\sqrt{p_{I|Y=200} (1-p_{I|Y=200})}.
\end{align*}
Similar to Example \ref{2stepexample4}, we observe that the difference between the two-step valuations relies mainly on
\begin{itemize}
	\item The risk premium $\kappa$ which reflects the difference between the real-world measure $\mathbb{P}$ and the risk-neutral measure $\mathbb{Q}$.
	\item The dependence between actuarial and financial risks (expressed as the difference between ${p_I}$ and $p_{I|Y=200}$ as well as the difference between $p_Y $ and $p_{Y|I=1} $).  $\hfill\blacktriangleleft$
\end{itemize}
\end{example}
\smallskip

We remark that in the literature, it is common to assume that financial and actuarial claims are independent (either under $\mathbb{P}$ or $\mathbb{Q}$\footnote{Note that independence under $\mathbb{P}$ does not necessarily imply independence under $\mathbb{Q}$, see \cite{dhaene2013dependence}.}). In that case, one can define a valuation that is MCV and ACV since the valuation is decoupled into two independent valuations, one for financial claims and one for actuarial claims. Even though extracting the $\mathbb{Q}$-dependence might be more complicated, different papers investigated the valuation under dependent financial and actuarial risks (see e.g. \citealp{liu2014generalized}, \citealp{deelstra2016role}, \citealp{zhao2018efficient}).

\addtocounter{example}{-1}
\begin{example}[continued]
	If we assume independence under $\mathbb{P}$ in the two-step financial valuation or independence under $\mathbb{Q}$ in the two-step actuarial valuation, both valuations lead to a fair valuation. This is in line with Lemma \ref{lemmaind}.
\end{example}

\section{Cost-of-capital valuation based on the two-step actuarial approach}\label{sectionNI}

\bluebis{Based on the two-step actuarial valuation introduced in Section \ref{ACvaluations}, we show how we can define a cost-of-capital valuation as a sum of a best estimate and a risk margin, different but in the same spirit as solvency regulations. Moreover, we illustrate such valuation on a portfolio of equity-linked life insurance contracts with dependent financial and actuarial risks in Section \ref{numericalapp}.}
 
\subsection{Best estimate, risk margin and cost-of-capital valuation}

\subsubsection{Best estimate}

In Article 77 of the DIRECTIVE 2009/138/EC (\cite{solvency2009directive}),
the best estimate is defined as the ``the probability-weighted average of
future cash-flows taking account of the time value of money'' (expected present
value of future cash-flows). Hence, the best estimate of an insurance liability can be
interpreted as an appropriate estimation of the expected present value based
on actual available information.

Based on our two-step actuarial valuation, we can define a broad notion of best estimate for a general claim $S$. Indeed, one can then generate stochastic actuarial scenarios, determine the financial price in each scenario and then average over the different scenarios. This leads to the following definition. 
\begin{definition}[Best estimate]
    For any claim $S \in L^{\infty}(\mathcal{F})$, the best estimate is given by 
    \begin{equation}
BE[S]=\mathbb{E}^{\mathbb{P}}\left[  \mathbb{E}^{\mathbb{Q}}\left[  S|\mathcal{F}^{(2)}\right]  \right]. \label{BE}%
\end{equation}
where $\mathbb{Q}$ is an absolutely continuous probability measure with respect to $\mathbb{P} $ conditional on $\mathcal{F}^{(2)}$.
\end{definition}

It turns out that the best estimate appears as a two-step actuarial valuation for which there is no distortion of the different measures, i.e. $ \pi^{(1)}\left[ S\right]=\mathbb{E}^\mathbb{Q}\left[S\right]$ and $ \pi^{(2)}\left[ S\right]=\mathbb{E}^\mathbb{P}\left[S\right] $. In general, the expression (\ref{BE}) could be hardly tractable since we can
possibly have an infinite number of actuarial scenarios. For practical
purposes, we will often consider the \textit{approximated} best estimate
$\widehat{BE}$ defined by%
\begin{equation}
\widehat{BE}[S]=\sum_{i=1}^{n}\mathbb{P}\left[
A_{i}\right]  \mathbb{E}^{\mathbb{Q}}\left[  \left.
S\right\vert A_{i}\right]  , \label{BEDEP2}%
\end{equation}
for a finite number $n$ of actuarial scenarios: $A_{1},A_{2},...,A_{n}%
\in\mathcal{F}^{\left(  2\right)  }$.

The best estimate in Equation \eqref{BE} appears as an average of
risk-neutral valuations which are applied to the risk which remains after
having conditioned on the actuarial filtration. Hereafter, we consider some
special cases:

\begin{itemize}
\item For any actuarial risk $S^{(2)}$, we
find that
\[
BE[ S^{(2)} ]=\mathbb{E}^{\mathbb{P}}\left[  S^{(2)}%
\right]  .
\]

\item For any product claim $S$ with independent actuarial and financial
risks (under $\mathbb{Q}$), we find that%
\begin{align*}
BE[S]  &  =\mathbb{E}^{\mathbb{P}}\left[  \mathbb{E}^{\mathbb{Q}}\left[  S^{(1)}\times S^{(2)}|\mathcal{F}%
^{(2)}\right]  \right] \\
&  = \mathbb{E}^{\mathbb{P}}\left[  S^{(2)}\times
\mathbb{E}^{\mathbb{Q}}\left[  S^{(1)}|\mathcal{F}%
^{(2)}\right]  \right] \\
&  =\mathbb{E}^{\mathbb{P}}\left[  S^{(2)}\times
\mathbb{E}^{\mathbb{Q}}\left[  S^{(1)}\right]  \right] \\
&  =\mathbb{E}^{\mathbb{P}}\left[  S^{(2)}\right]
\times \mathbb{E}^{\mathbb{Q}}\left[  S^{(1)}\right]. \numberthis \label{productformula}
\end{align*}
Hence, the actuarial risk is priced via real-world expectation and the financial risk via risk-neutral expectation. In fact, for the Equation \eqref{productformula} to hold, it is sufficient that the financial claim $S^{(1)}$ is independent from the actuarial filtration $\mathcal{F}^{(2)}$.
\end{itemize}

\subsubsection{Risk margin}

In order to motivate the risk margin, we recall that the best estimate is centered around the risk%
\[
\mathbb{E}^{\mathbb{Q}}\left[  S|\mathcal{F}%
^{(2)}\right]  .
\]
This risk represents the risk-neutral financial price of $S$ conditional on the actuarial information. Looking at the tail of this (actuarial) risk will provide information on the actuarial scenarios which yield the worst financial price. Hence, applying an actuarial valuation on this conditional financial price allows to measure the impact of the actuarial uncertainty on the risk-neutral price. This motivates the following definition.

\begin{definition}[SCR for actuarial risk]\label{defscr} For any claim $S \in L^{\infty}(\mathcal{F})$ and any actuarial valuation $\pi^{(2)}$, the SCR for actuarial risk is given by
	\begin{equation}
	SCR\left[S \right] =\pi^{(2)} \left[  \mathbb{E}^\mathbb{Q} \left[  S|\mathcal{F}^{(2)}\right]  \right]  -\mathbb{E}^{\mathbb{P}}\left[  \mathbb{E}^{\mathbb{Q}  }\left[  S|\mathcal{F}^{(2)}\right]  \right]  . \label{SCR1}%
	\end{equation}
	
\end{definition}
It turns out that the SCR appears as a two-step actuarial valuation for which we deducted the best estimate. If the valuation is coherent, thanks to the representation theorem (see Theorem \ref{dual}), the SCR for actuarial risk can be represented as 
\begin{equation*}
SCR\left[S \right]= \sup_{\tilde{\mathbb{P}}} \left\lbrace   \mathbb{E}^{\tilde{\mathbb{P}}}  \left[  \mathbb{E}^{\mathbb{Q}  }\left[  S|\mathcal{F}^{(2)}\right] \right]                 -\mathbb{E}^{\mathbb{P}}\left[  \mathbb{E}^{\mathbb{Q} }\left[  S|\mathcal{F}^{(2)}\right] \right]     \right\rbrace  
\end{equation*}
where the supremum is taken over a set of probability measures $\tilde{\mathbb{P}}$ absolutely continuous to $\mathbb{P}$. Hence, the SCR for actuarial risk can be interpreted as a worst case scenario: we can consider a family of stressed actuarial models (e.g. different mortality dynamics) and define the SCR as the value under the worst-case model.

\subsubsection{Cost-of-capital valuation of insurance liabilities}\label{RMFV}

In Solvency\ II, the fair value of insurance liabilities is defined as the sum
of the best estimate and the risk margin in which the latter is defined as the
cost of capital needed to cover the unhedgeable risks. 

In the spirit of regulatory directives, we define a cost-of-capital valuation (CoC valuation) based on our two-step actuarial valuation. This 
CoC valuation is defined as the sum of the best estimate (expected present
value) plus the risk margin (cost to cover unhedgeable risks) where the latter represents the cost of providing the SCR for actuarial risk. \bluebis{For a general overview of the cost-of-capital approach, we refer to \cite{wuthrich2013financial} and \cite{mohr_2011}.}

\begin{definition}[Cost-of-capital valuation]\label{FVDEF}
	For any claim $S \in L^{\infty}(\mathcal{F})$ and any actuarial valuation $\pi^{(2)}$, the cost-of-capital value of $S$ is defined by 
	\begin{equation}\label{fairvaluereg}
	\rho\left[S \right]  =BE\left[S \right]+iSCR\left[S \right] 
	\end{equation}
	with%
	\begin{align*}
	    BE[S]&=\mathbb{E}^{\mathbb{P}}\left[  \mathbb{E}^{\mathbb{Q}}\left[  S|\mathcal{F}^{(2)}\right]  \right] \\
	    SCR\left[S \right] &=\pi^{(2)}\left[  \mathbb{E}^\mathbb{Q} \left[  S|\mathcal{F}^{(2)}\right]  \right]  -\mathbb{E}^{\mathbb{P}}\left[  \mathbb{E}^{\mathbb{Q}  }\left[  S|\mathcal{F}^{(2)}\right]  \right],
	    \end{align*}
	where $i$ is the cost-of-capital rate and $SCR$ stands for the SCR
	for actuarial risk.
\end{definition}

\subsection{Numerical application: Portfolio of GMMB contracts}\label{numericalapp}

In this subsection, we show how to determine the cost-of-capital value \eqref{fairvaluereg} for a portfolio of guaranteed minimum maturity benefit (GMMB) contracts underwritten at time
$0$ on $l_{x}$ persons of age $x$. In
particular, we detail the numerical procedure for the best estimate and the SCR
for actuarial risk. Moreover,
we compare the fair valuation with the setting of
\cite{brennan1976pricing} in which complete diversification of mortality is assumed.

The GMMB contract offers at maturity the
greater of a minimum guarantee $K$ and a stock value if the policyholder is
still alive at that time. Let $T_{i}$ be the remaining lifetime of insured
$i,$ $i=1,2,\ldots,l_{x},$ at contract initiation. The payoff per policy can be
written as%
\begin{equation}
S=\frac{L_{x+T}}{l_{x}}\times\max\left(  Y^{(1)}(T),K\right)  , \label{gmmb}%
\end{equation}
with%
\[
L_{x+T}=\sum_{i=1}^{l_{x}}1_{\left\{  T_{i}>T\right\}  }.
\]
Here, $L_{x+T}$ is the number of policyholders who survived up to time $T$ and
$Y^{(1)}(T)$ is the value of the stock at time $T$. 

We consider a continuous time setting for the stock and the force of mortality dynamics. Let us assume that the dynamics of the stock process and the population
force of mortality are given by
\begin{align}
dY^{(1)}(t)  &  =Y^{(1)}(t)\left(  \mu dt+\sigma dW_{1}(t)\right) \label{ds}\\
d\lambda(t)  &  =c\lambda(t)dt+\xi dW_{2}(t), \label{dlambda}%
\end{align}
with $c,\xi,\mu$ and $\sigma$ are positive constants, and $W_{1}(t)=\rho
W_{2}(t)+\sqrt{1-\rho^{2}}Z(t)$. Here, $W_{2}(t)$ and $Z(t)$ are independent
standard Brownian motions. The specification of a non-mean reverting
Ornstein-Uhlenbeck (OU) process (\ref{dlambda}) for the mortality model allows
negative mortality rates. However, \cite{luciano2008mortality} and
\cite{luciano2017} showed that the probability of negative mortality rates is
 negligible with calibrated parameters. The benefit of such
specification is to allow tractability of mortality rates. Indeed, under
Equation (\ref{dlambda}), $\lambda(t)$ is a Gaussian process and $\int_{0}%
^{T}\lambda(t)dv$ is normal distributed.

\subsubsection{Best-estimate computation}

Since we want to determine the best estimate mortality, we assume that there is no risk premium in the actuarial market or,
equivalently, that Equation (\ref{dlambda}) holds under $\mathbb{P}$ and
$\mathbb{Q}$. Therefore, the calibration of the mortality intensity is
performed by estimating its dynamic under the real-world measure, and then
using it under the risk-neutral measure.\footnote{A similar approach is considered in \cite{luciano2017}} For the stock process, we define%
\[
dW_{1}^{\mathbb{Q}}(t)=\frac{\mu-r}{\sigma}dt+dW_{1}(t),
\]
where $\frac{\mu-r}{\sigma}$ represents the market price of equity risk. We
can then write the dynamics under $\mathbb{Q}$ as follows%
\begin{align}
dY^{(1)}(t)  &  =Y^{(1)}(t)\left(  rdt+\sigma dW_{1}^{\mathbb{Q}}(t)\right) \\
d\lambda(t)  &  =c\lambda(t)dt+\xi dW_{2}^{\mathbb{Q}}(t).
\end{align}
The best estimate for the aggregate payoff (\ref{gmmb}) is given by%
\[
BE[S] =\mathbb{E}^{\mathbb{P}}\left[  \mathbb{E}^{\mathbb{Q}}\left[
e^{-rT}\frac{L_{x+T}}{l_{x}}\times\max\left(  Y^{(1)}(T),K\right)  |\mathcal{F}%
^{(2)}\right]  \right]  .
\]
Under the independence assumption between the force of mortality and the stock
dynamics, one can easily show that the best estimate simplifies into%
\begin{align}
BE[S]   &  =\mathbb{E}^{\mathbb{P}}\left[\frac{L_{x+T}}{l_{x}}\right]  \times
\mathbb{E}^{\mathbb{Q}}\left[  e^{-rT}\max\left(  Y^{(1)}(T),K\right)  \right]
\\
&  =\text{ }_{T}p_{x}\left[  Y^{(1)}(0)N(d_{1})+Ke^{-rT}\left(
1-N(d_{2})\right)  \right] \\
&  =\text{ }\underbrace{\mathbb{E}^{\mathbb{P}}\left[  e^{-\int_{0}%
^{T}\lambda(v)dv}\right]  }_{_{T}p_{x}}\left[  Y^{(1)}(0)N(d_{1}%
)+Ke^{-rT}\left(  1-N(d_{2})\right)  \right]  \label{BEind}%
\end{align}
with%
\begin{align*}
d_{1}  &  =\frac{\ln\left(  \frac{Y^{(1)}(0)}{K}\right)  +(r+\frac{\sigma^{2}%
}{2})T}{\sigma\sqrt{T}},\\
d_{2}  &  =d_{1}-\sigma\sqrt{T}.
\end{align*}
We remark that the survival probability $_{T}p_{x}$ can be obtained in
closed-form (for details, see for instance \cite{mamon2004three}):%
\[
_{T}p_{x}=\mathbb{E}^{\mathbb{P}}\left[  e^{-\int_{0}^{T}\lambda(v)dv}\right]
=e^{A\lambda(0)+\frac{B}{2}},
\]
with
\begin{align}
A  &  =\frac{1}{c}\left(  1-e^{cT)}\right) \nonumber\\
B  &  =\frac{\xi^{2}}{c^{3}}\left(  cT+\frac{3}{2}-2e^{cT}+\frac{1}{2}%
e^{2cT}\right)  .
\end{align}
Under the dependence assumption, we provide in the next proposition the
approximated best estimate for the portfolio of GMMB contracts.

\begin{proposition}\label{prop1TS}
\label{prop1}If we denote by $_{T}p_{x}^{i}$ $(i=1,...,n)$ the survival rates
for each actuarial scenario\footnote{We assume that the actuarial scenarios are generated by a Monte-Carlo sample of i.i.d. observations}, the approximated best estimate for the aggregate payoff of
GMMB contracts:%
\[
\widehat{BE}[S]= \frac{1}{n} \sum_{i=1}^{n}  \mathbb{E}^{\mathbb{Q}}\left[  \left.  e^{-rT}%
\frac{L_{x+T}}{l_{x}}\times\max\left(  Y^{(1)}(T),K\right)  \right\vert L_{x+T}=l_{x}\text{
}_{T}p_{x}^{i}\right]
\]
is given by
\begin{equation}
\widehat{BE}[S]=\frac{1}{n} \sum_{i=1}^{n}\text{ }_{T}p_{x}^{i}
\left(\widetilde{Y}^{(1)}(0)N(d_{1})+e^{-rT}K\left(  1-N(d_{2})\right)  \right)  ,
\label{BEprop1}%
\end{equation}
with%
\begin{align*}
\widetilde{Y}^{(1)}(0)  &  =Y^{(1)}(0)e^{\frac{-\sigma\rho_{0}\sqrt{T}}%
{\sqrt{\frac{1}{2c}e^{2cT}-\frac{2}{c}e^{cT}+T+\frac{3}{2c}}}\left(  \frac
{c}{\xi}\ln\text{ }_{T}p_{x}^{i}+\frac{\lambda(0)}{\xi}\left(  e^{cT}%
-1\right)  \right)  }e^{-\frac{1}{2}\sigma^{2}\rho_{0}^{2}T},\\
\rho_{0}  &  =\frac{\rho\left(  \frac{1}{c}e^{cT}-\frac{1}{c}-T\right)
}{\sqrt{T\left(  \frac{1}{2c}e^{2cT}-\frac{2}{c}e^{cT}+T+\frac{3}{2c}\right)
}},\\
d_{1}  &  =\frac{\ln\left(  \frac{\widetilde{Y}^{(1)}(0)}{K}\right)  +\left(
r+\frac{1}{2}\sigma^{2}(1-\rho_{0}^{2})\right)  T}{\sigma\sqrt{\left(
1-\rho_{0}^{2}\right)  T}},\\
d_{2}  &  =d_{1}-\sigma\sqrt{\left(  1-\rho_{0}^{2}\right)  T}.
\end{align*}

\end{proposition}

\begin{proof}
The proof based on classical arguments of stochastic calculus can be found in
Appendix \ref{appendix1}.
\end{proof}

The approximated best estimate (\ref{BEprop1}) appears as an average of
Black-Scholes call option prices which are adjusted for the dependence between
the population force of mortality and the stock price processes. In each call
option, there is an adjustement of the current stock price $Y^{(1)}(0)$ to
$\widetilde{Y}^{(1)}(0)$, taking into account the realized survival rate
$_{T}p_{x}^{i}$ in each actuarial scenario. It is also worth noticing that in
case of independence ($\rho=0$), the approximated best estimate (\ref{BEprop1}%
) converges to the best estimate (\ref{BEind}).

To determine the best estimate \eqref{BEprop1}, we only need to generate
survival rates $_{T}p_{x}^{i}$ ($i=1,...,n$) and plug them into the
Black-Scholes option pricing formulas. Since the force of mortality dynamics
is given by%
\[
d\lambda(t)=c\lambda(t)dt+\xi dW_{2}(t),
\]
one can prove (for details, see Appendix \ref{appendix1}) that%
\[
\ln\text{ }_{T}p_{x}=-\int_{0}^{T}\lambda(s)ds\sim N\left(  \frac{\lambda
(0)}{c}\left(  e^{cT}-1\right)  ,\frac{\xi^{2}}{c^{3}}\left(  \frac{1}%
{2}e^{2cT}-2e^{cT}+cT+\frac{3}{2}\right)  \right)  .
\]
We generate $n=100000$ mortality paths. The benchmark parameters for the stock
and the force of mortality are given in Table \ref{table1}. The mortality
parameters follow from \cite{luciano2017} while the financial parameters are
based on \cite{bernard2016semi}. The mortality parameters correspond to UK
male individuals who are aged $55$ at time $0.$%

\begin{table}[h] \centering
\begin{tabular}
[c]{|l|}\hline
Parameter set for numerical analysis\\\hline
Force of mortality model: $c=0.0750,\xi=0.000597,\lambda(0)=0.0087.$\\
Financial model: $r=0.02,T=10,Y^{(1)}(0)=1,K=1,\sigma=0.2.$\\\hline
\end{tabular}
\caption{Parameter values used in the numerical illustration.}\label{table1}%
\end{table}%

\begin{table}[h]
	\centering
\begin{tabular}{@{}rrrr@{}}
\toprule
$\rho$ & Best estimate & $\rho$ & Best estimate \\ \midrule
-1.0 & 1.01132             & 0 & 1.00667            \\
-0.9 & 1.01086            & 0.1 & 1.00618           \\
 -0.8 & 1.01041              & 0.2 & 1.00568           \\
   -0.7 & 1.00995              &  0.3 & 1.00517              \\
 -0.6 & 1.00950            &  0.4 & 1.00466           \\
   -0.5 & 1.00904            &   0.5 & 1.00414           \\
 -0.4 & 1.00858            &  0.6 & 1.00360           \\
 -0.3 & 1.00811            &   0.7 & 1.00307             \\
 -0.2 & 1.00764            &   0.8 & 1.00252            \\
   -0.1 & 1.00716            &  0.9 & 1.00196            \\
   &        &   1.0 & 1.00141           \\ \bottomrule
\end{tabular}
	\caption{Best estimate for the GMMB contract using Equation \eqref{BEprop1}.}\label{tab:1}
\end{table}

Table \ref{tab:1} displays the best estimate per policy obtained using Equation (\ref{BEprop1}) for a range of correlation coefficients: $\rho\in
\lbrack-1,1]$. We observe that the best estimate slightly decreases with the increase
of the correlation parameter. This can be justified by a compensation effect
between the mortality and the stock dynamics:

\begin{itemize}
\item In case of positive dependence, high mortality scenarios (respectively
low mortality scenarios) are linked with high stock values (respectively low
stock values). In consequence, the expected value of the claim%
\[
S=\frac{L_{x+T}}{l_{x}}\times\max\left(  Y^{(1)}(T),K\right)
\]
will be reduced since high values of survivals $L_{x+T}$ will be associated
with low financial guarantees, $\max\left(  Y^{(1)}(T),K\right)  $, and vice-versa.

\item On the other hand, in case of negative dependence, high survival rates
will be linked with high financial guarantees, which implies a higher
uncertainty and an increase of the best estimate.
\end{itemize}

\subsubsection{Cost-of-capital value computation}

The cost-of-capital value of the insurance liability $S$ is then determined by%
\begin{equation*}
\rho\left[S \right]  =BE\left[S \right]+i SCR\left[S \right],
\end{equation*}
where the cost-of-capital rate $i$ is fixed at $6\%$ and the SCR for actuarial risk is given by
$$
S C R[S]=\pi^{(2)}\left[\mathbb{E}^{\mathbb{Q}}\left[S \mid \mathcal{F}^{(2)}\right]\right]-\mathbb{E}^{\mathbb{P}}\left[\mathbb{E}^{\mathbb{Q}}\left[S \mid \mathcal{F}^{(2)}\right]\right]
$$
for some coherent actuarial valuation $\pi^{(2)}$. For this numerical illustration, we consider the TVaR measure with a confidence level $p=$ $0.95$.

Figure \ref{FVGRAPH} represents the CoC value of $S$ for a range of
correlation coefficients: $\rho\in[-1,1]$. Overall, we observe an increase of the CoC value
of the GMMB contract under dependent mortality and equity risks. However, this effect is less pronounced for positive dependence. By comparison, the fair value of
$S$ under the assumption that mortality can be completely diversified (denoted
by $\rho^{\text{B-S}}$ for \cite{brennan1976pricing}), is given by Equation (\ref{BEind}):%
\begin{align}
\rho^{\text{B-S}}\left[S\right] &  =\mathbb{E}^\mathbb{P} \left[\frac{L_{x+T}}{l_{x}}\right]  \times
\mathbb{E}^{\mathbb{Q}}\left[  e^{-rT}\max\left(  Y^{(1)}(T),K\right)  \right]
\\
&  =\text{ }_{T}p_{x}\left[  Y^{(1)}(0)N(d_{1})+Ke^{-rT}\left(
1-N(d_{2})\right)  \right] \\
&  =1.0067.\nonumber
\end{align}

\begin{figure}[h]
	\begin{center}
		\includegraphics[width=0.8\textwidth]{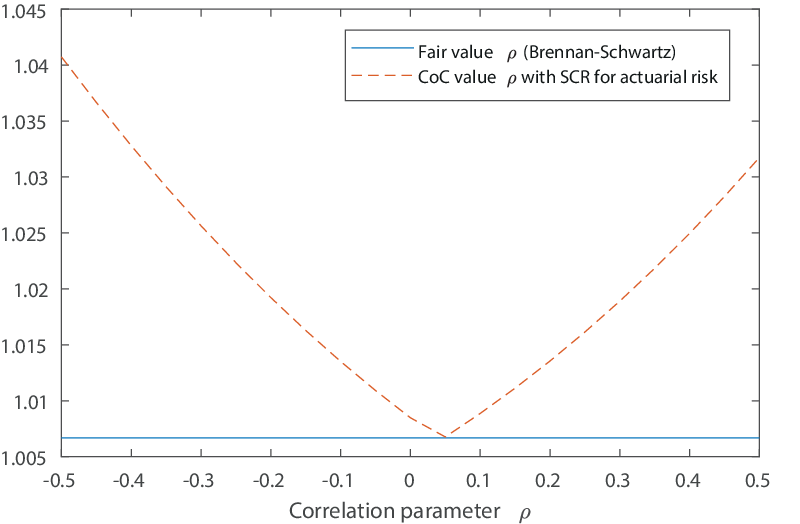}
		\caption{Comparison between the Cost-of-capital value for the GMMB contract under the two-step actuarial approach and the fair value of \cite{brennan1976pricing}.}%
		\label{FVGRAPH}%
	\end{center}
\end{figure}

From Figure \ref{FVGRAPH}, we clearly observe that this assumption underestimates the fair value of the contract since it does not take into account the actuarial uncertainty and the possible dependence with the financial market.

\section{Concluding remarks}\label{Conclu4}

In this paper, we have proposed a general actuarial-consistent valuation for insurance liabilities based on a two-step actuarial valuation. Actuarial-consistency requires that traditional actuarial valuation based on diversification applies to all actuarial risks. We have shown that every two-step actuarial valuation is actuarial-consistent and in the coherent setting, any actuarial-consistent valuation has a two-step actuarial valuation representation. We also studied under which conditions it is feasible to define a valuation that is actuarial-consistent and market-consistent. In general, it is not possible and the valuator should decide whether the valuation is driven by current market prices or historical actuarial information. 

Based on our two-step actuarial valuation, we have defined a cost-of-capital valuation in which the valuation is defined as the sum of a best estimate (expected value) and a risk margin (cost of providing the SCR for actuarial risks). The detailed numerical illustration has shown the important impact on risk management when
relaxing the independence assumption between actuarial and financial risks. In an extended B-S financial market, we determined the cost-of-capital value of a GMMB contract under dependent financial and actuarial risks. It turns out that the dependence structure has an important impact on the fair valuation and the related SCR.

As pointed out by \cite{liu2014generalized}, Solvency
II Directive strongly recommends the testing of capital adequacy requirements on the assumption of mutual dependence between financial markets and life insurance markets. In that respect, we believe that our two-step framework provides a plausible setting for the valuation of insurance liabilities with dependent financial and actuarial risks.

\section*{Acknowledgements}

Karim Barigou acknowledges the financial support of the Research Foundation - Flanders (FWO) (PhD funding) and the Joint Research Initiative on ``Mortality Modeling and Surveillance” funded by AXA Research Fund (postdoc funding). The authors would also like to thank Jan Dhaene from KU Leuven for useful discussions and helpful comments on this manuscript. \blue Finally, we sincerely thank the editor and anonymous referees for their pertinent remarks that significantly improved the quality of the manuscript.\black

\bibliographystyle{agsm}
\bibliography{bibliography2}

\newpage

\appendix

\section{Appendix: Proof of Proposition \ref{prop1TS}}\label{appendix1}

\begin{proof}
We recall that the dynamics of the stock process and the population force of
mortality under $\mathbb{Q}$ are given by%
\begin{align}
dY^{(1)}(t)  &  =Y^{(1)}(t)\left(  rdt+\sigma dW_{1}(t)\right) \label{app1}\\
d\lambda(t)  &  =c\lambda(t)dt+\xi dW_{2}(t) \label{app2}%
\end{align}
with $c,\xi,\mu$ and $\sigma_{1}$ are positive constants, and $W_{1}(t)=\rho
W_{2}(t)+\sqrt{1-\rho^{2}}Z(t)$. Here, $W_{2}(t)$ and $Z(t)$ are independent
standard Brownian motions under $\mathbb{Q}$. From (\ref{app2}), we note that%
\begin{align*}
d\left(  e^{-ct}\lambda(t)\right)   &  =-ce^{-ct}\lambda(t)dt+e^{-ct}%
d\lambda(t)\\
&  =\xi e^{-ct}dW_{2}(t).
\end{align*}
Hence, the force of mortality is a Gaussian process:
\[
\lambda(t)=\lambda(0)e^{ct}+\xi\int_{0}^{t}e^{-c(u-t)}dW_{2}(u).
\]
Moreover, we find that%
\begin{align*}
\int_{0}^{T}\lambda(s)ds  &  =\frac{\lambda(0)}{c}\left(  e^{cT}-1\right)
+\xi\int_{0}^{T}\int_{0}^{s}e^{-c(u-s)}dW_{2}(u)ds\\
&  =\frac{\lambda(0)}{c}\left(  e^{cT}-1\right)  +\xi\int_{0}^{T}\int
_{u}^{T}e^{-c(u-s)}dsdW_{2}(u)\\
&  =\frac{\lambda(0)}{c}\left(  e^{cT}-1\right)  +\frac{\xi}{c}\int_{0}%
^{T}\left(  e^{-c(u-T)}-1\right)  dW_{2}(u)\\
&  =\frac{\lambda(0)}{c}\left(  e^{cT}-1\right)  +\frac{\xi}{c}X_{T},
\end{align*}
with%
\[
X_{T}=\int_{0}^{T}\left(  e^{-c(u-T)}-1\right)  dW_{2}(u)\sim N\left(
0,\frac{1}{2c}e^{2cT}-\frac{2}{c}e^{cT}+T+\frac{3}{2c}\right)  .
\]
We can also remark that%
\begin{align*}
\mathbb{E}\left(  W_{1}(T)X_{T}\right)   &  =\mathbb{E}\left(  \int_{0}^{T}dW_{1}(u)\int_{0}%
^{T}\left(  e^{-c(u-T)}-1\right)  dW_{2}(u)\right) \\
&  =\rho\left(  \frac{1}{c}e^{cT}-\frac{1}{c}-T\right)  ,
\end{align*}
which leads to
\[
corr\left(  W_{1}(T),X_{T}\right)  =\frac{\rho\left(  \frac{1}{c}e^{cT}%
-\frac{1}{c}-T\right)  }{\sqrt{T\left(  \frac{1}{2c}e^{2cT}-\frac{2}{c}%
e^{cT}+T+\frac{3}{2c}\right)  }}\equiv\rho_{0}.
\]
We can then assume that
\[
W_{1}(T)=\frac{\rho_{0}\sqrt{T}}{\sqrt{\frac{1}{2c}e^{2cT}-\frac{2}{c}%
e^{cT}+T+\frac{3}{2c}}}X_{T}+\sqrt{T\left(  1-\rho_{0}^{2}\right)  }Z,
\]
where $Z$ is a standard normal r.v. independent of $X_{T}$.\newline
From%
\[
e^{-\int_{0}^{T}\lambda(s)ds}=\text{ }_{T}p_{x}^{i},
\]
we find that
\[
X_{T}=-\frac{c}{\xi}\ln\text{ }_{T}p_{x}^{i}-\frac{\lambda(0)}{\xi}\left(
e^{cT}-1\right)  .
\]
The stock price at time $T$ can be written as
\begin{align*}
Y^{(1)}(T)  &  =Y^{(1)}(0)e^{(r-\frac{1}{2}\sigma^{2})T+\sigma W_{1}(T)}\\
&  =Y^{(1)}(0)e^{\frac{-\sigma\rho_{0}\sqrt{T}}{\sqrt{\frac{1}{2c}%
e^{2cT}-\frac{2}{c}e^{cT}+T+\frac{3}{2c}}}\left(  \frac{c}{\xi}\ln\text{ }%
_{T}p_{x}^{i}+\frac{\lambda(0)}{\xi}\left(  e^{cT}-1\right)  \right)
}e^{(r-\frac{1}{2}\sigma^{2})T+\sigma\sqrt{1-\rho_{0}^{2}}\sqrt{T}Z}\\
&  =\widetilde{S}^{(1)}(0)e^{\left(r-\frac{1}{2}\sigma^{2}\left(  1-\rho_{0}%
	^{2}\right)\right)  T+\sigma\sqrt{1-\rho_{0}^{2}}\sqrt{T}Z},
\end{align*}
with
\[
\widetilde{S}^{(1)}(0)=Y^{(1)}(0)e^{\frac{-\sigma\rho_{0}\sqrt{T}}{\sqrt
{\frac{1}{2c}e^{2cT}-\frac{2}{c}e^{cT}+T+\frac{3}{2c}}}\left(  \frac{c}{\xi
}\ln \text{ }_{T}p_{x}^{i}+\frac{\lambda(0)}{\xi}\left(  e^{cT}-1\right)  \right)
}e^{-\frac{1}{2}\sigma^{2}\rho_{0}^{2}T}\text{.}%
\]
Finally, we find that
\begin{align*}
&  \mathbb{\mathbb{E}}^{\mathbb{Q}}\left[  \left. e^{-rT} L_{x+T}\times\max\left(
Y^{(1)}(T),K\right)  \right\vert e^{-\int_{0}^{T}\lambda(s)ds}=\text{ }%
_{T}p_{x}^{i}\right] \\
&  =l_{x}\text{ }_{T}p_{x}^{i}\mathbb{E}^{\mathbb{Q}}\left[  \left.
 e^{-rT}K+ e^{-rT}\max\left(  Y^{(1)}(T)-K,0\right)  \right\vert e^{-\int_{0}^{T}\lambda
(s)ds}=\text{ }_{T}p_{x}^{i}\right] \\
&  =l_{x}\text{ }_{T}p_{x}^{i}\left(  \widetilde{S}^{(1)}(0)N(d_{1})+e^{-rT}K\left(
1-N(d_{2})\right)  \right)  ,
\end{align*}
which ends the proof.
\end{proof}

\end{document}